\documentclass{IEEEtaes}

\usepackage{color,array,amsthm}
\usepackage{graphicx}

\usepackage{subcaption}
\usepackage[utf8]{inputenc}
\usepackage{cite}
\usepackage{amsmath,amssymb,amsfonts}
\usepackage{steinmetz}
\usepackage{algorithmicx}
\usepackage{algorithm}
\usepackage{algpseudocode}
\usepackage{graphicx}
\usepackage{textcomp}
\usepackage{algpseudocode} 
\usepackage{mathtools}
\usepackage{dsfont}
\usepackage[dvipsnames]{xcolor}
\usepackage{nicefrac}
\usepackage{txfonts}
\usepackage{fdsymbol}
\usepackage{mdframed}
\usepackage{todonotes}
\usepackage{tikz}
\usepackage{eso-pic} 

\usepackage{cleveref}
\DeclareMathOperator*{\argmin}{arg\,min}

\DeclareMathOperator{\EX}{\mathbb{E}}
\DeclareMathOperator{\atan2}{atan2}
\newcommand{\rt}{\rho_{_{\rm T}}}
\newcommand{\ra}{\rho_{_{\rm A}}}
\newcommand{\ta}{\theta_{\rm A}}

\newcommand{\tao}{\theta_{\rm A0}}

\newcommand{\xd}{\mathbf{x}_{\rm D}}
\newcommand{\xa}{\mathbf{x}_{\rm A}}
\newcommand{\xdo}{\mathbf{x}_{\rm D0}}
\newcommand{\xao}{\mathbf{x}_{\rm A0}}
\newcommand{\xc}{\mathbf{x}_{\rm C}}
\newcommand{\thetae}{\theta_{\rm eng}}
\newcommand{\taue}{\tau_{\rm eng}}
\newcommand{\phie}{\phi_{\rm eng}}
\newcommand{\reng}{r_{\rm eng}}
\newcommand{\p}{\textbf{p}}
\newcommand{\Se}{S_{\rm engage}}
\newcommand{\R}{\mathbb{R}^2}
\newcommand{\uvec}{\hat{\mathbf{u}}}
\newcommand{\x}{\mathbf{x}}
\newcommand{\ro}{r_{_{\rm T}}}

\newcommand{\q}{\textbf{q}}

\newcommand{\xcap}{\x_{\rm capture}}
\newcommand{\lcap}{\Lambda_{\rm capture}}
\newcommand{\lscap}{\Lambda^*_{\rm capture}}
\newtheorem{lemma}{Lemma}

\newtheorem{definition}{Definition}

\newtheorem{remark}{Remark}
\newtheorem{assumption}{Assumption}
\newcommand{\re}{\color{red}}
\newcommand{\dd}{\mathrm{d}}
\makeatletter
\def\ps@IEEEtitlepagestyle{\ps@headings}

\makeatother
\begin{document}

\title{Multi-Attacker Single-Defender Target Defense in Conical Environments}

\author{Arman Pourghorban}
\affil{University of North Carolina at Charlotte, NC, 28223, USA} 

\author{Dipankar Maity}
\member{Senior Member, IEEE}
\affil{University of North Carolina at Charlotte, NC, 28223, USA}



\corresp{
}

\authoraddress{Arman Pourghorban is with the University of North Carolina at Charlotte, Charlotte, NC 28223 USA 
(e-mail: \href{mailto:apourgho@charlotte.edu}{apourgho@charlotte.edu}). Dipankar Maity is with the University of North Carolina at Charlotte, Charlotte, NC 28223 USA 
(e-mail: \href{mailto:dmaity@charlotte.edu}{dmaity@charlotte.edu}).}

\maketitle
\makeatletter
\def\@copyrightspace{}%
\def\@IEEEpubid{}%
\def\@IEEEcopyright{}%
\makeatother
\makeatletter
\def\ps@headings{%
   \def\@oddhead{}%
   \def\@evenhead{}%
   \def\@oddfoot{}%
   \def\@evenfoot{}%
}
\def\@IEEEpubid{}
\pagestyle{headings}
\thispagestyle{headings}
\makeatother

\begin{abstract} We consider a variant of the target defense
problem in a planar conical environment where a single defender is tasked to capture a
sequence of incoming attackers.  The attackers’ objective is
to breach the target boundary without being captured by
the defender. As soon as the current attacker breaches the
target or gets captured by the defender, the next attacker appears at the boundary of the environment and moves radially toward the target with maximum speed.  Therefore, the defender’s final location at the end
of the current game becomes its initial location for the next
game. The attackers pick strategies that are advantageous
for the current as well as for future engagements between the defender and the remaining attackers. 
The attackers  have their own sensors with limited range, using which they can perfectly detect if the defender is within their sensing range. 
 We derive
equilibrium strategies for all the players to optimize the
capture percentage using the notions of \textit{capture distribution}.  Finally, the theoretical results are verified through numerical examples using Monte-Carlo type random trials of
experiments.
\end{abstract}


\section{INTRODUCTION}
T{\scshape arget} defense  games in conical environments are critical in strategic defense due to their prevalence in various natural and man-made settings, such as mountainous regions, valleys, and specific architectural structures.  These environments present unique challenges and opportunities for defense and have found applications in surveillance \cite{lewin1979conic}, patrolling \cite{melikyan1991simple}, and guarding \cite{bajaj2024multi}. The emergence of guarding a target in real-world applications
has motivated a large amount of research in this area \cite{lee2022vision,ho2022game,eloycoop,garcia2020barrier,s2020airport,hoorfar2023securing}. Notably, several of these works explicitly incorporate environments shaped by conic sections \cite{bajaj2022competitive}, highlighting the practical and theoretical significance of this geometric structure in real-world defense applications.

We study a perimeter-defense game between a defender and an attacker team in a conical environment. The attacker team sends its team-members sequentially to breach the target. On the other hand, a \textit{single} defender is tasked to guard a target by capturing the attackers. 
Perimeter defense games, first introduced in \cite{isaacs1999differential}, are a variant of pursuit-evasion problems in which a team of defenders is tasked with capturing attackers before they breach a protected perimeter \cite{velhal2022decentralized}. Much of the existing literature addresses this problem using the Hamilton-Jacobi-Isaacs (HJI) framework or as a one-sided optimal control problem \cite{lee2021guarding, chen2016multiplayer, garcia2019strategies, chen2014path}. However, due to the curse of dimensionality, these approaches are limited to low-dimensional settings and relatively simple objectives.

In contrast, our game involves a sequence of agents engaging with a single opponent, a structure that is not readily amenable to traditional HJI or optimal control formulations, as we will explain soon. 
While some alternative methods—such as dimensionality reduction techniques \cite{festa2013decomposition}, Voronoi partitioning \cite{selvakumar2016evasion}, and reachable set analysis \cite{sun2017pursuit}—have been proposed to mitigate scalability and handle more complex objectives, they still do not address the key challenges posed by our setting.

In this work, we leverage key properties of the Apollonius Circle, particularly the recently developed robust pursuit strategies \cite{dorothy2021one}, and extend them to account for information asymmetry and sensing limitations in our setting.
\par
 \subsection{Key Related Works}
A prevalent assumption evident in the previously cited works and the majority of the literature \cite{khrenov2021geometric,deng2023multiple,wang2024target} pertains to the availability of full-state information. Specifically, it is assumed that all agents participating in the game possess access to the states of all other agents. This assumption holds pivotal significance as it facilitates the application of differential-game techniques for the derivation of equilibrium strategies. However, the acquisition of full-state information proves unfeasible in numerous realistic security scenarios and one of the key challenges to our problem.
\par
Of particular relevance to this paper is the work in \cite{bajaj2022competitive}, which considers a scenario where attackers appear randomly at arbitrary time instants along the boundary of a conical environment.
While our problem shares the feature of random attacker appearances on the boundary, it differs in that the attackers arrive sequentially over time in our problem rather than in a purely random fashion as in \cite{bajaj2022competitive}. This sequential structure enables the design of strategic solutions with analytical guarantees, whereas deriving optimality guarantees for random arrival times remains intractable.

To address this intractability, \cite{bajaj2022competitive} adopts a competitive analysis framework, evaluating the performance of various online defender algorithms against arbitrary input sequences, relative to an optimal offline algorithm with full knowledge of future events. They design and analyze three online algorithms and identify parameter regimes in which these algorithms achieve finite competitive ratios—providing some measure of performance guarantees. However, their work does not consider any sensing capabilities for the attackers, effectively limiting the complexity of evasive strategies that attackers can employ.

Similarly, works such as \cite{adler2022rolenew} and \cite{macharet2020adaptive} assume simple, fixed a priori strategies for attackers attempting to breach a target perimeter. These models also neglect attacker sensing capabilities, precluding evasive behavior and reducing the problem to straightforward interception rather than dynamic pursuit-evasion.
\par

In our prior work we consider a similar problem for circular environment \cite{pourghorban2022target} and also for periodic arrivals \cite{pourghorban2023target}. 
However, as it turns out, addressing the problem for a conical segment has unique challenges and opportunities, specially due to the side boundaries of the conic being impenetrable and treated as obstacles. 
Thus, in the limit when the cone angle becomes 180$^o$, the cone does not become a circle and the agents have to go all the way around. 
Consequently, majority of the tools developed in our prior work requires re-investigation and a new analysis.  

In our prior work \cite{pourghorban2022target}, we studied a target-defense problem in a circular environment where a sequence of incoming attackers is tasked to breach a circular target while a single defender is tasked to guard the target. We analytically quantified the capture percentage for both finite and infinite sequences of incoming attackers. In \cite{pourghorban2023targetnew}, which was a natural extension of \cite{pourghorban2022target}, we studied
the same problem with the condition where each attacker knows
the entry point of the last attacker and this information is used
to find an optimal entry point for the next attacker. We extended this work \cite{pourghorban2023target} by considering a sequence of periodically incoming attackers. We proposed three algorithmic approaches as the defender's strategy and analytically computed a lower bound on one of the proposed strategies. 
However, all of these works considered a circular environment, which does not directly solve the problem when restricted to a conical environment.
\subsection{Main Contributions}
 The agents (defender and the attackers) have limited sensing regions. Therefore, they may not have information about their opponents all the time. This sensing limitation enforces them to
consider both \textit{positional} and \textit{sensing} advantages while
choosing their maneuvering strategies. We study
this problem by decomposing each one-vs-one attacker-defender game into
two phases: \textit{partial} and \textit{full information phases}. Each one vs-one game starts with the \textit{partial information phase} and
then either it proceeds to the \textit{full information phase} or the
game terminates in the breach of the target. The objective
of the defender (attackers) is to maximize (minimize) the percentage of captured attackers.
\par
 The main contributions of this paper are as follows:
 \begin{itemize}
     \item We formulate a sequential target-defense game in a conical environment where each attacker has limited sensing and acts based on partial information, and the defender must account for long-term outcomes across engagements.
     \item We derive the optimal strategies for all the agents by studying a max-min Stackelberg game.
     \item We present analytical performance bounds for the game by introducing an  approximation of the agents' optimal strategy.
 \end{itemize}
 \subsection{Organization}
The remainder of this paper is organized as follows. 
Section~\ref{sec:Problemformulation} formulates the problem, states the parametric assumptions, and presents useful definitions along with necessary background material. 
Section~\ref{sec:GameAnalysis} provides the game analysis, including the description of the problem phases. 
Section~\ref{sec:OptEng} derives the optimal engagement configurations for the agents. 
Section~\ref{sec:capprob} analyzes the capture probability. 
Section~\ref{sec:GameProgress} discusses the game progression across multiple engagements. 
Section~\ref{sec:bound} establishes analytical bounds on the defender’s performance. 
Section~\ref{sec:Simu} presents simulation and numerical results validating the theoretical developments. 
Finally, Section~\ref{sec:Conclusion} concludes the paper.


\par
\textit{Notation:} 
All vectors are denoted with lowercase bold symbols, e.g., $\mathbf{x}$. 
We use $\uvec(\theta)$ to denote the unit vector $[\cos\theta,~\sin\theta]^\intercal$.
For a non-zero vector $\x$, we use $\hat{\x}$ to denote the unit vector along $\x$, and use $\angle \x$ to denote the angle $\atan2(\x[2],\,\x[1])$, where $\x[i]$ denotes the $i$-th component of vector $\x$.


\section{Problem Statement} \label{sec:Problemformulation}
We consider the problem of guarding a convex conical target region. 
Let $\mathcal{R}_T$ $= \{\x \in  \R   \mid \| \x\| \le \ro, \ | \angle \x | \le \Phi\}$ denote the target region, where $\ro$ and $\Phi$ denote the radius and angular span of the conical target.
The target is equipped with a sensing annulus of radius $\rt$, as shown in Fig.~\ref{fig:introPic}. We will refer to this region as the Target Sensing Region (TSR).
Together the target and the TSR form a conical game environment $\mathcal{R}_E$ $= \{\x \in  \R  \mid \| \x\| \le \ro+\rt, \ | \angle \x | \le \Phi\}$.
The two side edges of the conical environment are considered to be impenetrable boundaries, and the attackers can neither enter nor exit through them. 

\begin{figure}
    \centering
    \includegraphics[width = 0.25 \textwidth]{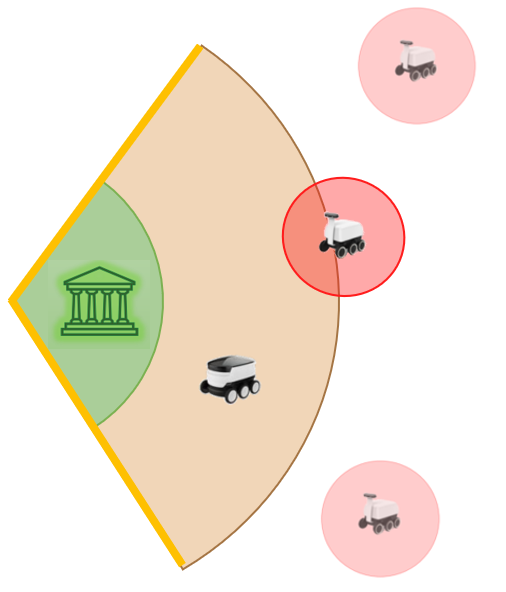}
    \put (-35, 50 )  {defender}
    \put (-55, 50 ) {$\longleftarrow$}
    \put (-5, 85 )  {attacker}
    \put (-22, 85 ) {$\longleftarrow$}
       \put (-84, 70 ) {\line(1,0) {44}}
    \put (-75, 75) {$\rt$}
              \put (-18, 128 ) 
           {\textcolor{black}{{\line(1,0) {10.5}}}}
 \put (-4, 128) {$\ra$}
           \put (-99, 39 ) 
           {\textcolor{black}{{\line(-0.65,1) {20.7}}}}
 \footnotesize{\put (-120, 50) {$\ro$}}

           \caption{\colorbox{green!30}{\textcolor{green!30}{g}}: Target region of radius $\ro$, \colorbox{red!20}{\textcolor{red!20}{g}}: attackers sensing region of radius $\ra$, \colorbox{orange!60!black!25}{\textcolor{orange!60!black!25}{g}}: Target Sensing Region (TSR) of radius $\rt$, \textcolor{orange!50!black!50}{\textbf{\Large --}} TSR Boundary,\textcolor{orange!90!black!40}{\textbf{\Large --}} Impenetrable boundary.}
    \label{fig:introPic}
    \vspace{-12 pt}
\end{figure}
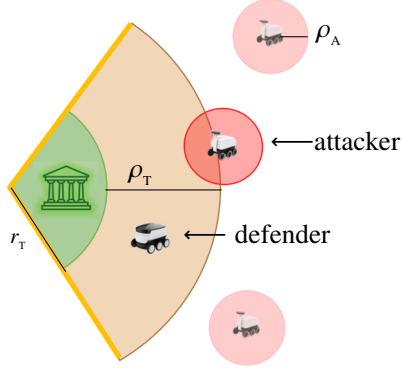

The target is defended by a single agent, hereafter called as the \textit{defender}, which can sense any attacker that is present within the TSR. 
The TSR is an `early warning system' that enables the defender to perfectly measure the location of any attacker within the TSR. 
The defender cannot sense any attacker outside of the TSR.
Therefore, the first moment the defender can sense an attacker is when the attacker appears on the TSR boundary.
We consider a random arrival pattern for the attackers, i.e., they do not appear from any fixed set of points on the TSR boundary. 
We assume that each attacker appears on the TSR boundary with a uniform random probability independent of the former arrivals, similar to the scenarios considered in prior works \cite{velhal2022decentralized, pourghorban2022target}.
This arrival pattern has the maximum \textit{entropy} (i.e., uncertainty). 

The game is \textit{sequential} in nature \cite{pourghorban2022target}, meaning only one attacker enters the TSR at any given time and the next attacker enters the TSR as soon as the current one breaches the target boundary or gets captured by the defender. 
This makes the problem challenging since the defender's final location in the current game becomes its initial location for the next game. 
Therefore, while engaging with the current attacker, the defender must consider the consequences of this engagement on all future arrivals of the attackers. 
We will consider both finite and infinite number of arrivals.


Each engagement may terminate in either a \textit{capture} of the attacker, or a \textit{breach} of the target, or an \textit{evasion} of the attacker.
Let $\xa(t),~\xd(t) \in \R$ denote the positions of a representative attacker and the defender at time $t$. 
In this work, we assume that $\xd(t) \in \mathcal{R}_E$ for all~$t$.

We say the attacker is \textit{captured} if there exists a time $t_f$ such that 
\begin{align*}
    \xa(t_f) \in \mathcal{R}_E, ~{\rm and} \quad  \|\xd(t_f) -\xa(t_f) \|=0.
\end{align*}
Likewise, a \textit{breach} happens if there exists a $t_f$ such that 
\begin{align*}
    \xa(t_f) \in \mathcal{R}_T, ~{\rm and}\quad  \|\xd(t_f) -\xa(t_f) \| > 0,
\end{align*}
and an \textit{evasion} happens if there exists a $t_f$ such that the attacker gets out of the TSR uncaptured, i.e., 
\begin{align*}
    \xa(t_f) \notin \mathcal{R}_E, ~{\rm and} \quad \|\xd(t) -\xa(t) \| > 0, ~~~\forall~t \le t_f.
\end{align*}

Let $S(N)$ denote the total number of attackers captured when a total of $N$ attackers have completed the game. 
Note that, due to the random arrival locations of the attackers and/or possibly due to any randomized motion tactics employed by the players, the number of captures (i.e., $S(N)$) is a random variable. 
 The defender's objective is to maximize the expected capture fraction
\[
 J \triangleq \frac{\EX [S(N)]}{N},
\]
 where the expectation is taken with respect to the randomness in the arrival patterns and the randomization (if any) in the agent's motion tactics.
For the case of an infinite number of attacker arrivals, we consider the asymptotic expected capture fraction
\[
 J_\infty \triangleq \liminf_{N \to \infty} \frac{\EX [S(N)]}{N}.
\]
This is a zero-sum game, and therefore, the attackers' objective is to minimize $J$ (or $J_\infty$).

\subsection{Agents' Dynamics and Capabilities}
The defender and the attackers are assumed to have first-order dynamics \cite{khrenov2021geometric,garcia2019strategies,lee2021guarding}, i.e., 
\begin{align} \label{eq:dyn}
    \dot{\mathbf{x}}_{\rm A} = v_{\rm A} \uvec(\psi_{\rm A}), \qquad \dot{\mathbf{x}}_{\rm D} = v_{\rm D} \uvec(\psi_{\rm D}),
\end{align}
where the defender (respectively, the attacker) directly controls its speed and heading angle by selecting $v_{\rm D}$ and $\psi_{\rm D}$ ($v_{\rm A}$ and $\psi_{\rm A}$), respectively.
Without any loss of generality, we assume that the defender and the attackers
have the speed limit of 1 and $\nu$, respectively, i.e., $|v_{\rm D}(t)|\le 1$ and $|v_{\rm A}(t)|\le \nu$ for all $t$.\footnote{
In case the maximum speeds are $v_{D,\max}$ and $v_{A,\max}$, we shall use scaling and time dilation to transform the system to obtain maximum speeds $1$ and $\nu$, respectively.  
To that end, we define $\nu =v_{A,\max}/v_{D,\max}$, $c =v_{D,\max}$, $\bar \x_i(t) = \x_i( t/c )$  for $i=A, D$. 
With this change of variables, we obtain $\dot{\bar{\x}}_i = \bar{\mathbf{v}}_i $ with  $\|\bar{\mathbf{v}}_D(t)\| \le 1$ and $\|\bar{\mathbf{v}}_A(t)\| \le \nu$.
} 
Furthermore, to avoid a trivial scenario we assume that $\nu < 1 $, i.e., the defender is faster.
\par
Each attacker is equipped with a sensor and is able to sense the defender only if the defender is within a distance of $\ra$ or less, i.e., the attackers have a circular sensing footprint of radius $\ra$; see Fig.~\ref{fig:introPic}.
After detecting the defender, an attacker can find the best breaching point on the target, or decide to get out of the TSR uncaptured if breaching is impossible but escape is possible, or evade for some time before getting captured by the defender when both breaching and evasion is impossible. 
 This evasive tactic is an important aspect for this game since it forces the defender to pursue and capture the attacker at a location that is likely to be unfavorable for the defender against the next attacker.

 The attacker may not sense the defender at the moment it appears on the TSR boundary. 
 In such cases, the attacker moves radially toward the target center until they sense the defender or breaches the target; i.e., the attacker takes the shortest path to the target in absence of any knowledge about the defender's location.


\subsection{Parametric Assumption}
   The overall outcome of
this game depends on the game
parameters $\ro,\ra $, $\rt $, $\nu $ and $\Phi$. 
In this paper we consider the following assumptions. 


\begin{assumption} \label{assm:deadlock}
    The parameters satisfy
    \begin{align}
        \frac{\rt}{\ra} \ge 1+\frac{2\nu}{(1-\nu^2)}.
    \label{eq:Asm2}
\end{align}
\end{assumption}
Prior work, such as \cite{shishika2021partial}, has shown that $\frac{\rt}{\ra} \ge \frac{2\nu}{(1-\nu^2)}$ is necessary to ensure that capture is a possible event in a circular environment. 
In other words, $\frac{\rt}{\ra} < \frac{2\nu}{(1-\nu^2)}$ is sufficient for the attackers to breach or evade, and consequently, it may result in a trivial game with $J = J_\infty = 0$. 
Here we assume a slightly stronger condition than $\frac{\rt}{\ra} \ge \frac{2\nu}{(1-\nu^2)}$ in \eqref{eq:Asm2} (i.e., we include the $+1$) to have a \textit{sufficient} condition that not only excludes the trivial scenario mentioned earlier, but also excludes certain \textit{deadlock} scenarios; see Lemma~\ref{lemma:sufficiency} for details.


\begin{assumption} \label{assm:small target}
    The parameters satisfy $\nu \ro\leq \rt$.
\end{assumption}
This assumption ensures that the defender can capture any attacker from the target center. 
We make this assumption only to simplify the exposition of this paper.

\subsection{Apollonius Circle}\label{sec:Apcircle}
\textit{Apollonius Circle} (AC) plays an important role in analyzing games with  dynamics \eqref{eq:dyn}. 
Given the locations of the defender and the attacker at time $t$, the AC contains all the points the attacker can reach before the defender gets there.
The set of all such points is a circular region  with center $\xc(t)$ and radius $r_{\rm C}(t)$ given by
\begin{align} \label{eq:AC}
    \xc(t) = \alpha\xa(t) - \beta \xd(t), \quad r_{\rm C}(t) = \gamma\|\xa(t) - \xd(t)\|,
\end{align}
where 
\begin{align} \label{eq:alpha_beta_gamma}
    \alpha = (1-\nu^2)^{-1}, \quad \gamma = \nu \alpha,\quad \beta = \nu\gamma.
\end{align}
From the definition of AC, one may immediately notice that any point inside the AC is reachable by the attacker without getting captured.
Recently, \cite{dorothy2021one} provided a formal proof that the attacker cannot get out of the AC without getting captured. 
The following lemma is the key result from \cite{dorothy2021one}, which will play an instrumental role in analyzing our game.

\begin{lemma}[\!\!\cite{dorothy2021one} Theorem 1] \label{lem:arXiv}
Let $\xd(t_0)$ and $\xa(t_0)$ denote the positions of the agents at a specific time $t_0$. Then, for any $\epsilon>0$, the defender will capture the attacker within a distance of $\epsilon$ from the \textit{Apollonius circle} constructed at time $t_0$ by following the strategy
\begin{align} \label{eq:pursuit}
    &v_{\rm D}(t) = 1, \nonumber \\
    &\uvec(\psi_{\rm D}(t)) = \frac{ \left(r_C(t_0) + \epsilon - r_C(t)\right) \hat{\x}_{AD}(t) + \nu \mathbf{y}(t) }{\| \left(r_C(t_0) + \epsilon - r_C(t)\right) \hat{\x}_{AD}(t) + \nu \mathbf{y}(t) \|}, 
\end{align}
for all $ t\ge t_0$, where $\hat{\x}_{AD}(t)$ is the unit vector in the direction $\xa(t)-\xd(t)$, and $\mathbf{y}(t) = \xc(t) - \xc(t_0)$.
\hfill $\triangle$
\end{lemma}
A pictorial representation is provided in Fig.~\ref{fig:arXivlemma}.
\begin{figure}     
    \centering
    \begin{tikzpicture}
        \filldraw[color=black, fill=white, dashed](-1,0) circle (1.55);
        \filldraw[color=black, fill=white](-1,0) circle (1.3);
        \filldraw[color=black, fill=white, dashed](3,0) circle (1.5);
        \filldraw[color=black!80, fill=white!5, very thick](3.4,-0.3) circle (0.85);
         \draw[->, color = black!50!white] (-0.95,-2) -- (-0.95,-1.05);
                  \draw[->, color = black!50!white] (3.25,-1.6) -- (3.33,-0.95);
                \draw[->] (3.0,0) -- (3.33,-0.23);
                \draw[->, color = black!50!white] (3.25,-1.6) -- (3.68,-1.90);
                \draw[->, color = black] (3.25,-1.6) -- (3.7,-1.30);
                \draw[-, color = green!80!blue,dashed] (3.7,-1.30) -- (4.1,-1);              
        \draw [latex reversed-latex reversed, line width=1pt](0.15,0) -- (0.71,0);
        \put (-30, -2.7) {\color{black} $\smallblackcircle$}
        \put (-30, -30) {\color{red} $\smallblackcircle$}
        \put (-30, -58) {\color{blue} $\smallblackcircle$}
        \put (94, -11.5) {\color{black} $\smallblackcircle$}
        \put (92, -27) {\color{red} $\smallblackcircle$}
        \put (89.5, -48) {\color{blue} $\smallblackcircle$}
         \put (10, 3) { \small $\epsilon$}
         \put (-53, -57) {\color{blue} {\small$\xd(t_0)$}}
         \put (-53, -26) {\color{red} {\small$\xa(t_0)$}}
         \put (-53, 0) {\color{black} {\small$\xc(t_0)$}}
         \put (-24, -50) {\small{$\hat{\x}_{AD}(t_0)$}}
         \put (82, -2.7) {\color{black} $\smallblackcircle$}
         \put (100, -11) {\color{black} {\small$\xc(t)$}}
         \put (93, 0) {\color{black} {\scriptsize$\mathbf{y}(t)$}}
         \put (67, -37) {\scriptsize{$\lambda_1\hat{\x}_{AD}(t)$}}
                  \put (81, -56) {\scriptsize{$\lambda_2\mathbf{y}(t)$}}
    \end{tikzpicture}
    \caption{Illustration of the pursuit strategy from Lemma~\ref{lem:arXiv}, where the defender (blue) intercepts the attacker (red) within an $\epsilon$-neighborhood of the Apollonius circle (dashed). The subfigure on the right shows the weighted combination of direction vectors used in the control law to guide the defender toward the interception point.}
    \label{fig:arXivlemma}
\end{figure}
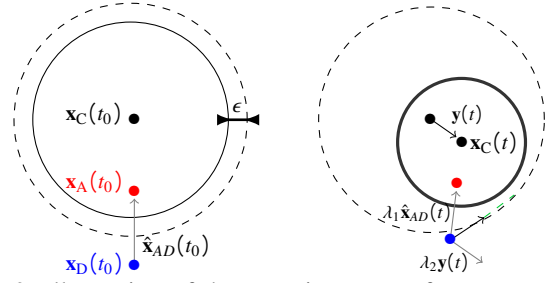

\begin{remark}
    The defender's strategy, as given in \eqref{eq:pursuit}, is in the form of $\lambda_1(t) \hat{\x}_{AD}(t) + \lambda_2(t) \mathbf{y} (t)$, where $\hat{\x}_{AD}$ represents the pure-pursuit element and $\mathbf{y}(t)$ represents the ``guard the critical point'' element. The pure-pursuit aspect aims to minimize the distance between the defender and the attacker, whereas the latter ensures the target is still being guarded while the defender is reducing the distance from the attacker. By effectively balancing these two elements, the defender can secure the capture without needing any prior knowledge of the attacker's motion strategy. In essence, strategy \eqref{eq:pursuit} offers a robust pursuit method that guarantees capture against \textit{all} attacker tactics.
\end{remark}

\section{Game Analysis} \label{sec:GameAnalysis}
As mentioned earlier, the sequential arrival of the attackers results in a sequence of coupled 1-vs-1 games, where the coupling comes from the fact that the defender's terminal location in one game becomes its initial location for the next one.
Each 1-vs-1 game can be divided into two game phases, namely, the \textit{Full information phase} and the \textit{Asymmetric information phase} \cite{pourghorban2022target,pourghorban2023target,pourghorban2023targetnew}. 
In the \textit{Full information phase}, both the defender and the attacker can sense each other i.e., $\xa(t) \in \mathcal{R}_E $ and $\|\xa(t)-\xd(t)\|\le \ra$.
On the other hand,  only the defender can sense the attacker in the \textit{Asymmetric information phase}, i.e., $\xa(t) \in \mathcal{R}_E $ and $\|\xa(t)-\xd(t)\|>\ra$.

The \textit{Asymmetric information phase} is a partial/incomplete information game for the attacker whereas it is a complete information game for the defender. 
The defender shall utilize this information asymmetry to its advantage. 
Note that the asymmetric information phase continues until the defender comes within the attacker's sensing region. 
Since the defender is faster, it can stay out of the evader's sensing region as long as it wants, and consequently, it may elongate the asymmetric information phase at the risk of letting the attacker breach the target. 
A key question that we will address in the remainder of this paper is ``\textit{when} and \textit{where} should the defender first come in contact with the attacker's sensing region?'' 

\begin{remark}
    For some 1-vs-1 games, there might not be an asymmetric information phase
    since the defender is sensed as soon as the attacker appears on the TSR boundary (i.e., the defender is too close to the TSR boundary point from where the attacker is entering). 
    However, we will show later (\Cref{lemma:sufficiency}) that such a scenario is avoided under \Cref{assm:deadlock} and the proposed optimal strategy of the defender. 
\end{remark}

Let $t$ be the time when an attacker first senses the defender, i.e., the full information phase starts. 
Consequently, at this moment, we have $\|\xa(t)-\xd(t)\| = \ra$, which we may rewrite as 
\begin{align} \label{eq:engagementNecessaryCondition}
\xd(t) - \xa(t) = \ra \uvec(\theta),
\end{align}
for some $\theta \in [0, 2\pi)$.
Therefore, the AC constructed at this time will have the center and radius of 
\begin{align}\label{eq:AC2}
\begin{split}
    \xc(t) &= \alpha \xa(t) - \beta \xd(t) = \xa(t) - \beta\ra \uvec(\theta), \\
    r_C(t) &= \gamma \| \xa(t) - \xd(t)\| = \gamma\ra,
    \end{split}
\end{align}
where we have used the identity $\alpha - \beta = 1$; see \eqref{eq:alpha_beta_gamma}. 

From the properties of the AC, we may conclude the outcome of this engagement as follows:
\begin{itemize}
    \item If the AC has an intersection with the target, then \textit{target breach} is inevitable and the attacker must pick a breach point and move directly toward it; see \Cref{fig:gameOutcome}(a).
    \item If the AC intersects the boundary of the TSR, then by Assumption \ref{assm:deadlock}, it cannot simultaneously intersect the target region. In this case, \textit{evasion} is inevitable, and the attacker must select an evasion point on the TSR boundary and move directly toward it; see \Cref{fig:gameOutcome}(b).
    \item Otherwise, if the AC neither has an intersection with the target nor with the outside of the TSR, then \textit{capture} is guaranteed by \Cref{lem:arXiv}; see \Cref{fig:gameOutcome}(c). 
    However, the attacker gets to decide the point on the AC where it will get captured. This leads to a Stackelberg games \cite{von2010market}.
\end{itemize}

\begin{figure}
    \centering
    \begin{subfigure}[b]{0.3\linewidth}
    \centering
        \includegraphics[trim = 150 190 880 130, clip, width= \linewidth]{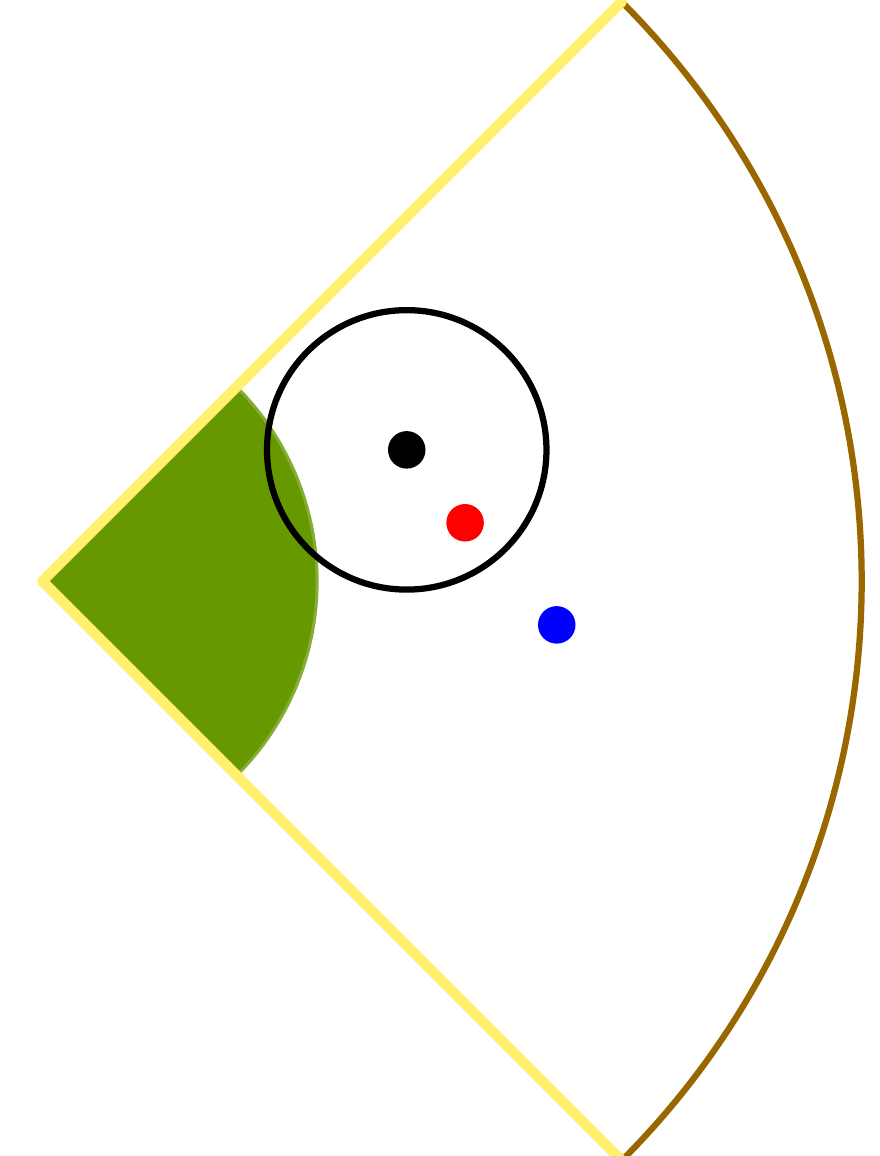}
        \caption{Breach}
    \end{subfigure}
    \begin{subfigure}[b]{0.3\linewidth}
    \centering
        \includegraphics[trim = 150 190 880 130, clip, width=  \linewidth]{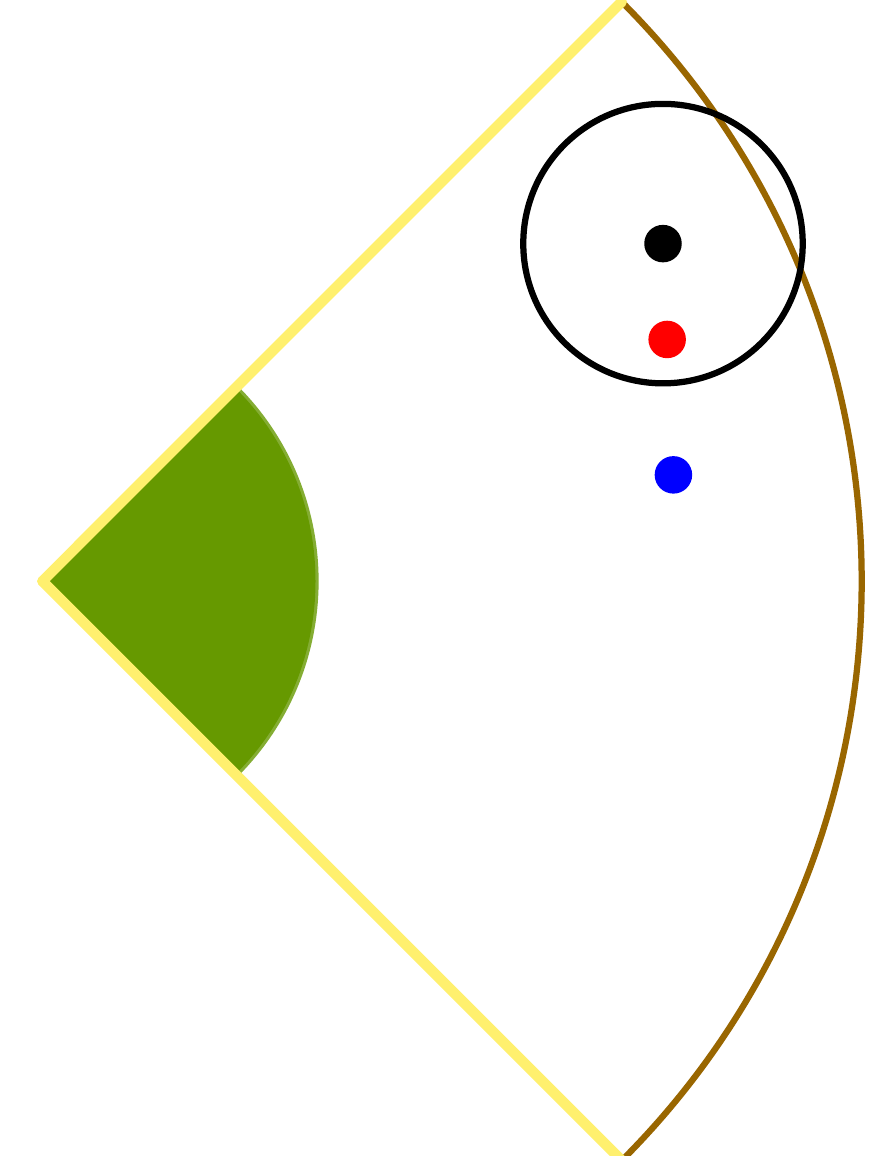}
        \caption{Evasion}
    \end{subfigure}
    \begin{subfigure}[b]{0.3\linewidth}
    \centering
        \includegraphics[trim = 150 190 880 130, clip, width=  \linewidth]{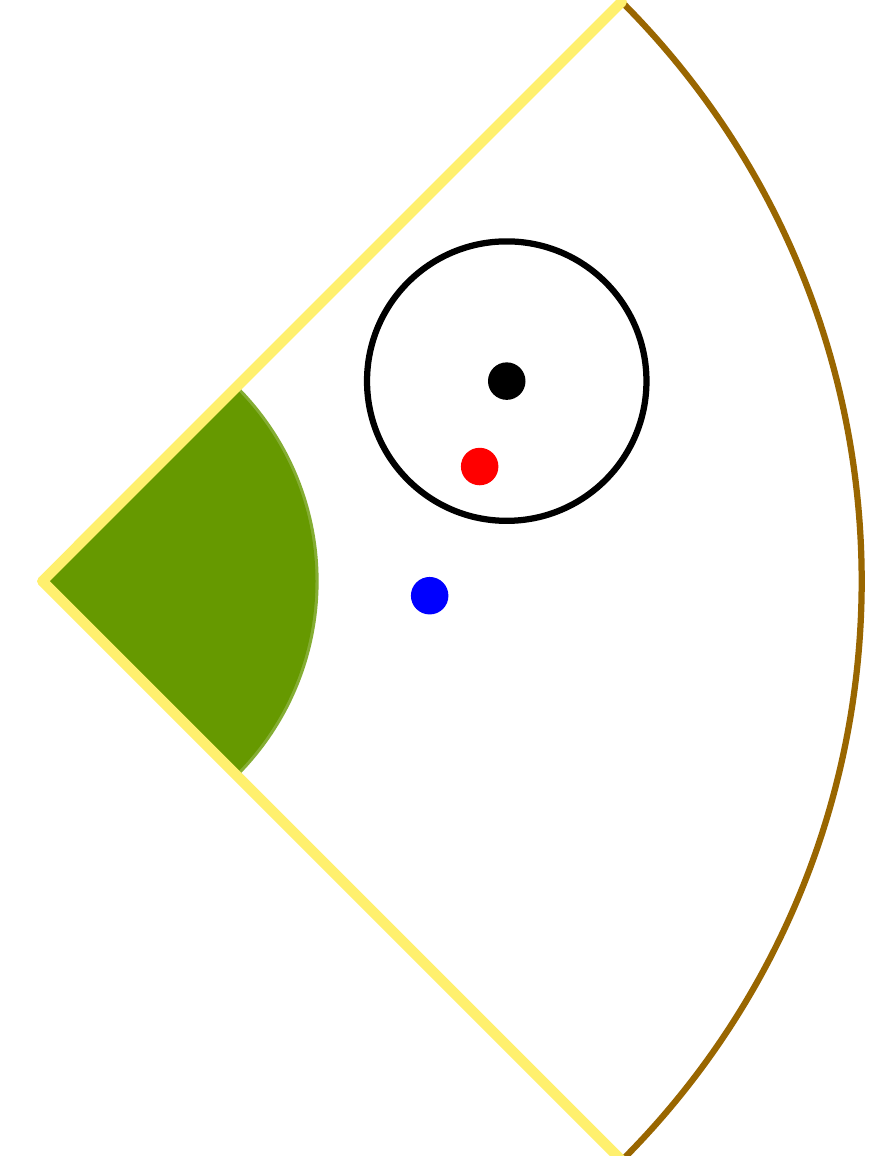}
        \caption{Capture}
    \end{subfigure}
    \caption{The blue and red dots represent the positions of the defender and
the attacker, respectively. The black
circles represent the Apollonius circle.}
    \label{fig:gameOutcome}
\end{figure}

The capture outcome is certainly the most interesting one since the capture location will become the defender's initial location for the next game. 
Therefore, the attacker must strategically choose to terminate this game in a location which is the most favorable for the future incoming attackers.
 The \textit{breach} and \textit{evasion} outcomes will be discussed later in \Cref{sec:GameProgress}.


\subsection{Equilibrium Strategy Pair for Capture}

As discussed in the previous section, the current attacker aims to get captured at a location such that it optimizes the probability of successful breach for the subsequent attackers. We approach the problem by first analyzing a simplified scenario involving only two consecutive arrivals, and then extend the derived strategies to the full sequence of attackers. This approach is justified by the assumption that the attackers’ arrival locations are independent and identically distributed, allowing generalization from the two-attacker case to the full sequence.
\par
Let $f(\x, \theta)$ denote the probability that a defender starting at $\x \in \mathcal{R}_E$ captures an attacker appearing on the TSR boundary at an angle $\theta$. 
Then, the total capture probability of a defender starting at $\x$ is
\begin{align*}
    P(\x) = \int_{-\Phi}^{\Phi} f(\x,\theta) p(\theta) \dd \theta, 
\end{align*}
where $p(\theta)$ is the probability that attacker appears at angle $\theta$.
Under the uniform arrival assumption, we have
$p(\theta)=\frac{1}{2 \Phi}$.
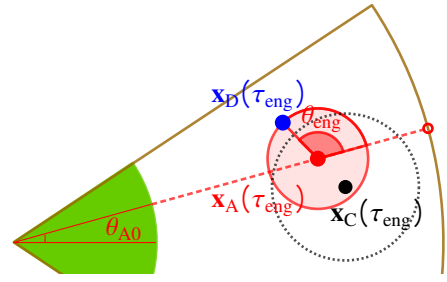
\begin{figure}
    \centering
    \centering
    \includegraphics[trim = 470 550 670 80, clip,width = 0.35 \textwidth]{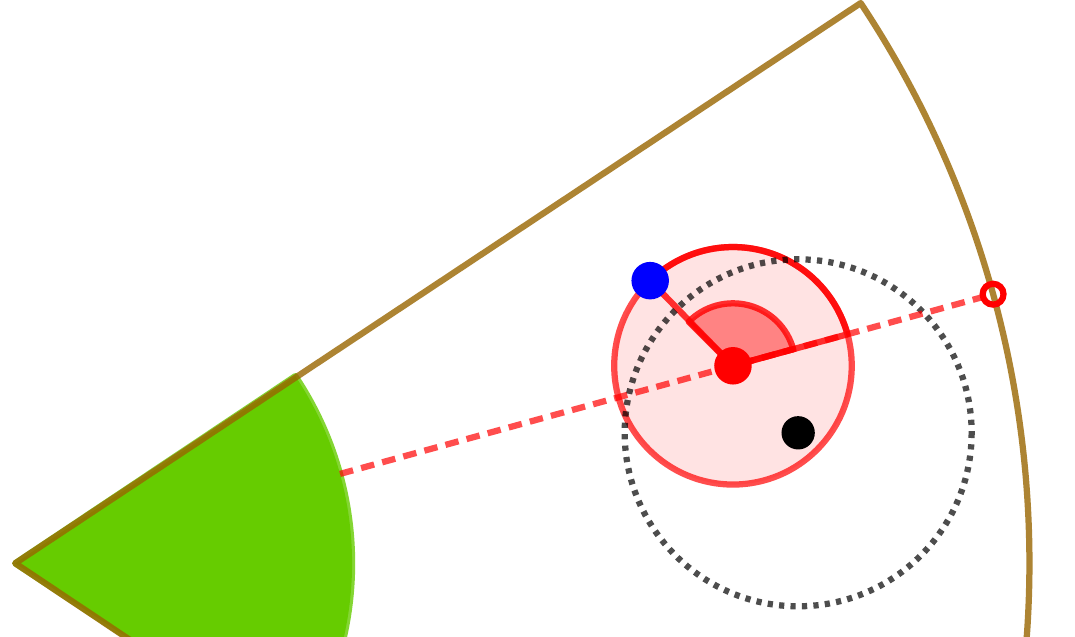}
 \put(-95,25) {$\re{\xa(\taue)}$}
 \put(-50,20) {${\xc(\taue)}$}
  \put(-95,65) {$\color{blue}{\xd(\taue)}$}
  \put(-61,57) {$\re{\thetae}$}
  \put (-170, 11.7 ) {\re{\line(3.6,1) {52}}}
  \put (-170, 11.7 ) {\re{\line(1,0) {54}}}
  \put (-158, 11.7 ) {\re{\line(0,1) {3}}}
  \put(-135,14.5) {$\re{\tao}$}
           \caption{An engagement configuration with the intruder at time $\taue$,  where {\re $\medbullet$ }  and  $\color{blue}{\medbullet}$  denote the locations of the intruder and the defender, respectively.}
    \label{fig:eng}
        \vspace{-12 pt}
\end{figure}
\noindent  Now, suppose an attacker appears on the TSR boundary at the location $(\ro+\rt)\uvec(\tao)$ and the defender starting at $\xdo$ is able to begin the full information phase at an angle $\thetae$ after a duration of $\taue$; see Fig.~\ref{fig:eng}. 
Recall that the attacker moves radially (i.e., in the shortest path) toward the target center until it senses the defender.
Therefore, at the moment the full information phase starts, the agents' locations are 
\begin{align} \label{eq:xaxd_taue}
\begin{split}
    \xa(\taue) & = (\ro+\rt - \nu \taue)\uvec(\tao), \\
    \xd(\taue) & = \xa(\taue) + \ra \uvec(\thetae).
\end{split}
\end{align}
From the engagement configuration we may construct the corresponding AC and conclude that the attacker shall get captured at the point $\xcap$ that minimizes the defender's probability for capturing the next attacker, i.e.,
\begin{align}\label{eq:capture_point}
    \xcap = \argmin_{\x \in \mathcal{A}(\taue, \thetae)} P(\x),
\end{align}
where $\mathcal{A}(\taue, \thetae)$ denotes the AC with its center and radius given by \eqref{eq:AC2} when we set $t = \taue$ and $\theta = \thetae$ in \eqref{eq:AC2}.
On the other hand, the defender's objective is to maximize the probability of the next capture while capturing the current one. 
Therefore, the defender must choose the engagement configuration $(\taue, \thetae)$ accordingly. 
In other words, the agents perform the following max-min optimization: 
\begin{align} \label{eq:stackelberg}
    \max_{(\taue, \thetae) \in \lcap(\xdo,\xao)}~\min_{\x \in \mathcal{A}(\taue, \thetae)} P(\x),
\end{align}
where $\lcap(\xdo, \xao)$ is the set of all engagement configurations that result in the capture of the current attacker.
This set depends on the initial locations $\xao$ and $\xdo$. 
In \Cref{sec:OptEng} we will discuss how to compute $\lcap(\xdo, \xao)$. 

\begin{remark}
    The $\max$-$\min$ optimization in \eqref{eq:stackelberg} is a Stackelberg game where the defender (i.e., the `leader') first picks the engagement point $(\taue,\thetae)$ and the attacker (i.e., the `follower') responds to the defender's choice.
    Once we have computed $\lcap(\xdo, \xao)$ and $P(\cdot)$, we may numerically solve \eqref{eq:stackelberg}. 
    However, a numerical solution requires an exhaustive search that is computationally challenging and potentially intractable, especially in real-time operations. 
    Therefore, in this paper we will instead focus on obtaining an analytical solution to a modified version of \eqref{eq:stackelberg} in the form:
    \begin{align} \label{eq:stackelberg2}
    \max_{(\taue, \thetae) \in \lscap(\xdo,\xao)}~\min_{\x \in \mathcal{A}(\taue, \thetae)} P(\x),
\end{align}
    where $\lscap \subset \lcap$ is to be discussed in the next section.
\end{remark}

\section{Optimal Engagement Configurations} \label{sec:OptEng}
In this section, we discuss the computation of $\lcap(\xdo, \xao)$ and $\lscap(\xdo, \xao)$. 
For a given engagement pair $(\taue, \thetae) \in \lcap(\xdo, \xao)$ the corresponding AC must not have any intersection with the target to prevent breach.
That is, we must have 
\begin{align} \label{eq:guard_condition}
    \begin{cases}
        \|\xc(\taue)\| \ge \ro+\gamma\ra, &\text{ if } \ | \angle \xc(\tau) | \le \Phi, \\
        \|\xc(\taue)- \ro\uvec(\pm \Phi)\| \ge \gamma\ra, &\text{ otherwise. } 
    \end{cases}
\end{align}
\begin{remark}
Throughout the above and subsequent equations, the notation $\pm \Phi$ 
is to be interpreted as follows: the positive case ($+\Phi$) corresponds 
to $\angle \xc(\tau) > \Phi$, while the negative case ($-\Phi$) 
corresponds to $\angle \xc(\tau) < -\Phi$.
\end{remark}

Equation \eqref{eq:guard_condition} is both necessary and sufficient to ensure that the target is guarded from the attacker (i.e., no breach).
However, \eqref{eq:guard_condition} does not guarantee the prevention of the attacker's escape. 
The necessary and sufficient condition to prevent escape is to ensure that the AC does not have any intersection with the outside of the TSR. 
That is,
\begin{align}\label{eq:escape_condition}
    \begin{cases}
        \|\xc(\taue)\|+ \gamma\ra \le r_0 + \rt,  ~ ~ ~ ~ ~ ~ ~ ~ ~\text{ if } \ | \angle \xc(\tau) | \le \Phi, \\
        \|\xc(\taue)- (\ro+\rt)\uvec(\pm \Phi)\| \le \gamma\ra,  ~ ~ ~ ~ ~  ~ ~ \text{ otherwise. } 
    \end{cases}
\end{align}
Therefore, $(\taue, \thetae) \in \lcap(\xdo, \xao)$ must satisfy both \eqref{eq:guard_condition} and \eqref{eq:escape_condition}. 
Next, we further simplify these two conditions to obtain a better geometrical understanding of $\lcap(\xdo, \xao)$ which will further assist in solving \eqref{eq:stackelberg2}. 
To that end, we use \eqref{eq:xaxd_taue} in \eqref{eq:guard_condition} to simplify $\|\xc(\taue)\| \ge \ro+\gamma\ra$ as
\begin{align} \label{eq:critical_theta}
    \sin^2\left(\frac{\thetae - \tao}{2}\right) \ge &  \frac{(\ro+\gamma\ra)^2 - (\ro+\rt-\taue \nu -\beta\ra)^2}{4\beta\ra(\ro+\rt-\taue \nu)}. 
\end{align} 
Similarly, by simplifying the second case in \eqref{eq:guard_condition} one may obtain 
\begin{align} \label{eq:critical_theta'}
\sin^2\left(\frac{\thetae - \tao}{2}\right) \ge  \frac{M-(\ro+\rt-\taue \nu -\beta\ra)^2}{4\beta\ra(\ro+\rt-\taue \nu)}\end{align}
    where $M$ is equal to
    \begin{align*}
            M  =& (\gamma \ra)^2 -\ro^2+2(\ro+\rt-\taue \nu)\ro \cos(\tao \pm \Phi)
            \\
            &-2\beta\ra\ro \cos(\thetae \pm \Phi).
    \end{align*}

\begin{remark} \label{rem:critical_tau}
If the right hand side of \eqref{eq:critical_theta} is negative, or equivalently $\taue \le (\rt - \nu\ra/(1-\nu)   )/\nu$, then any $\thetae \in [-\pi,\pi]$ will satisfy the inequality in \eqref{eq:critical_theta}.
On the other hand, if the right hand side of \eqref{eq:critical_theta} is greater than 1, or equivalently, $\taue \ge (\rt -\nu\ra/(1+\nu))/\nu$, then no $\thetae$ satisfies \eqref{eq:critical_theta}. 
This provides two critical times 
\begin{subequations} \label{eq:critical_times}
\begin{align}
    \tau_1^* = \frac{\rt}{\nu} - \frac{\ra}{(1-\nu)},\\ \tau_2^* = \frac{\rt}{\nu} - \frac{\ra}{(1+\nu)},
\end{align}
\end{subequations}
with the properties that if the defender can reach the sensing boundary of the attacker before $\tau_1^*$, then guarding is guaranteed regardless of the angle $\thetae$ at which the engagement starts, whereas if the defender cannot reach the sensing boundary before $\tau_2^*$, then breaching is inevitable regardless of the engagement angle.
(A similar conclusion may also be drawn from \eqref{eq:critical_theta'}, however, the expressions of these critical times will be dependent on $\tao$ and much more complicate than the ones in \eqref{eq:critical_times}.) 
\end{remark}

For every engagement time $\taue$, one may solve \eqref{eq:critical_theta} to obtain the set $\Theta(\taue, \tao)$ such that every $\thetae \in \Theta(\taue, \tao)$ satisfies \eqref{eq:critical_theta}.
As stated in \Cref{rem:critical_tau}, $\Theta(\taue)$ can become an empty set or the entire range $[0,2\pi)$ depending on $\taue$. 

Likewise, we also use \eqref{eq:xaxd_taue} in \eqref{eq:escape_condition} to simplify $\|\xc(\taue)\|+ \gamma\ra \le \ro + \rt$ to
{\small
\begin{align} \label{eq:critical_theta_guard}
\sin^2\left(\frac{\thetae - \tao}{2}\right) \le 
&\; \frac{
(\ro + \rt - \gamma\ra)^2 
- (\ro+\rt-\taue \nu -\beta\ra)^2
}{
4\beta\ra(\ro+\rt-\taue \nu)
}
\end{align}
} 
\noindent 
and rewrite $ \|\xc(\taue)- (\ro+\rt)\uvec(\pm \Phi)\| \le \gamma\ra$ as 
\begin{align} \label{eq:critical_theta_guard'}
\sin^2\left(\frac{\thetae - \tao}{2}\right) \le  \frac{\hat{M}-(\ro+\rt-\taue \nu -\beta\ra)^2}{4\beta\ra(\ro+\rt-\taue \nu)}
\end{align}
    where $\hat{M}$  equals to
    \begin{align*}
    \hat{M} ={}& (\gamma \ra)^2 - (\ro + \rt)^2 \\
    & + 2(\ro + \rt - \taue \nu)(\ro + \rt) \cos(\tao \pm \Phi) \\
    & - 2\beta \ra \ro \cos(\thetae \pm \Phi).
\end{align*}

\begin{remark}
If the right hand side of \eqref{eq:critical_theta_guard} is negative, i.e., $\tau < \ra/(1+\nu)$, then no $\thetae$ satisfies \eqref{eq:critical_theta_guard}. 
On the other hand, that quantity is greater than $1$, i.e., $\tau > \ra/(1 - \nu)$, then every $\thetae \in [-\pi, \pi]$ satisfies \eqref{eq:critical_theta_guard}. 
This provides two critical times for escape prevention
\begin{subequations}
\label{eq:critical_times_guard}
\begin{align}
    \tau_3^* =\frac{\ra}{(1+ \nu)},\\ 
    \tau_4^* = \frac{\ra}{(1 - \nu)},
\end{align}
\end{subequations}
with the properties that if the defender reaches the sensing boundary of the attacker before $\tau_3^*$, then escape is guaranteed (assuming that the attacker acts rationally), whereas if the defender starts engagement after $\tau_4^*$, then escape is impossible without being captured. (A similar conclusion may also be drawn from \eqref{eq:critical_theta_guard'}, however, the expressions of these critical times will be dependent on $\tao$ and much more complicate than the ones in \eqref{eq:critical_times_guard}.) 
\end{remark}

Equations \eqref{eq:critical_theta}-\eqref{eq:critical_theta'} provide the feasible engagement points on the attacker's sensing boundary so that the engagement results in \emph{guaranteed guarding} of the target. 
Solving \eqref{eq:critical_theta} or \eqref{eq:critical_theta'} (if $|\angle \xc(\taue)| \le \Phi$ the defender uses \eqref{eq:critical_theta}; otherwise, it uses \eqref{eq:critical_theta'}) with equality provides $\theta_{\text{eng}}(\tau)$ such that, for any $\pi \ge |\theta| > \theta^1_{\text{eng}}(\tau)$, equations \eqref{eq:critical_theta} or \eqref{eq:critical_theta'} are satisfied.
Together, $\xa(\tau)$ and $\theta^1_{\text{eng}}(\tau)$ provide the \emph{guaranteed guarding} engagement surface $S_g(\tau)$ at time $\tau$
\begin{align*}
    S_g(\tau) \triangleq \{\x : \x = \xa(\tau)+\rho_A[c_\theta, s_\theta], |\theta| \in [\theta_{\text{eng}}^1(\tau), \pi]\}.
\end{align*}
\par
A representation of $S_g(\tau)$ is shown with a blue arc in the left subplot in Fig.~\ref{fig:engagementset}.
Starting the \textit{Full Information Phase} at time $\tau$ on the surface $S_g(\tau)$ ensures guaranteed guarding, however, it may also lead to evasion of the attacker. 

Analogously, the \emph{guaranteed escape prevention} engagement surface at time $\tau$ is described by
\begin{align*}
    S_c(\tau)= \{\x : \x = \xa(\tau)+\rho_A[c_\theta, s_\theta], |\theta| \le \theta^2_{\text{eng}}(\tau)\},
\end{align*}
where $\theta^2_{\text{eng}}(\tau)$ is the solution of \eqref{eq:critical_theta_guard} or \eqref{eq:critical_theta_guard'} (if $|\angle \xc(\taue)| \le \Phi$ the defender uses \eqref{eq:critical_theta_guard}; otherwise, it uses \eqref{eq:critical_theta_guard'}) with equality.
An illustration of this engagement surface/arc is shown in red in the left subplot of Fig.~\ref{fig:engagementset}.
\begin{figure}     
    \centering
    \includegraphics[trim = 60 40 30 20, clip, width=0.46\linewidth]{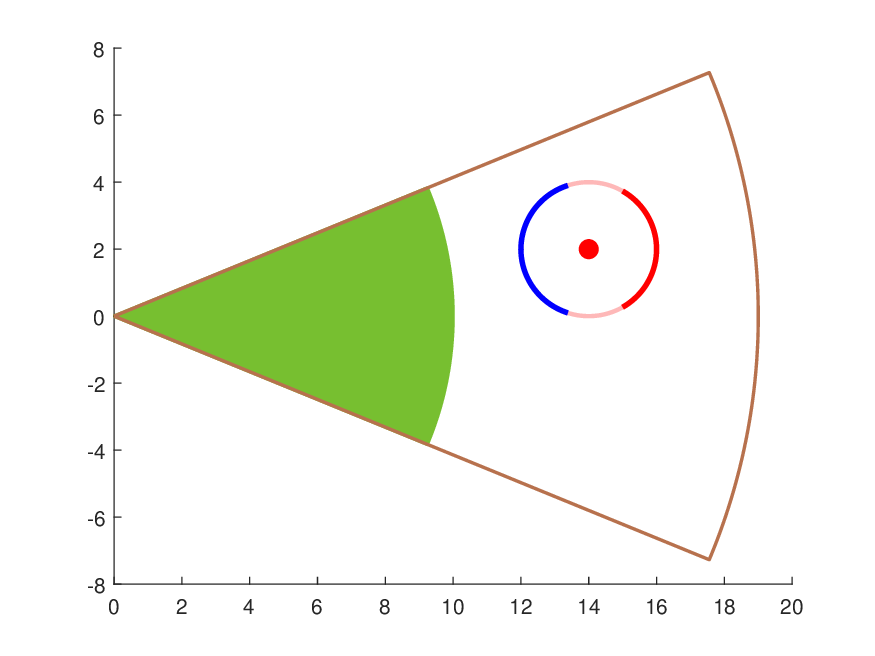}
        \centering
    \includegraphics[trim = 200 270 200 200, clip, width=0.4 \linewidth]{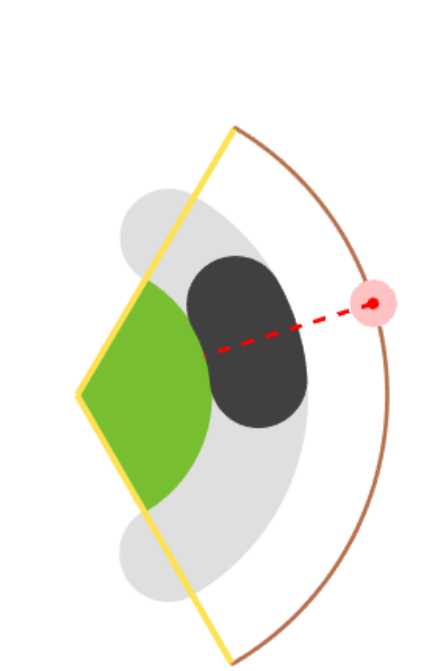}
    \caption{Left: Engagement set $S_g(\tau)$ (blue arc) and $S_c(\tau)$ (red arc) on the attacker's sensing boundary. 
Right: The light gray region denotes all possible capture locations when the defender follows the strategy described in Section \ref{sec:GameAnalysis}.\ref{sec:simplified}. 
For the given red attacker position, the dark gray region represents the subset of locations where this attacker can be captured.}
    \label{fig:engagementset}
\end{figure}
\begin{remark} \label{rem:time}
If the engagement phase begins at time $\tau$, then a necessary condition to prevent escape and ensure guarding is that the engagement starts at the intersection $S_g(\tau) \cap S_c(\tau)$.  
In terms of the critical times $\tau^*_i,\ i = 1,\ldots,4$, the necessary condition for successful guarding without escape is $\tau_2^* \ge \tau_4^*$, which simplifies to $\rt \ge 2\gamma\ra$.

\end{remark}
\begin{definition}
(Engagement Set) This is the collection of the pairs $(\taue,\thetae)$ that satisfy \eqref{eq:critical_theta} or \eqref{eq:critical_theta'} {\rm and} \eqref{eq:critical_theta_guard} or \eqref{eq:critical_theta_guard'}. 
That is, 
\begin{equation} \label{eq:SE}
\Se =
\left\{ (\thetae, \taue)~:~\thetae, \taue ~ \text{ satisfy } 
\begin{aligned}
&\eqref{eq:critical_theta} ~ or ~ \eqref{eq:critical_theta'},\\
&\eqref{eq:critical_theta_guard} ~ or ~ \eqref{eq:critical_theta_guard'} \ 
\end{aligned}
\right\}.
\end{equation}

\end{definition}

\noindent If the \textit{full information phase}—where both agents can see each other—begins on the \textit{Engagement Set} $\Se$ defined in \eqref{eq:SE}, then the attacker can neither \textit{escape} nor \textit{breach}, and thus \textit{capture} is guaranteed.

\begin{figure}
    \centering
    \includegraphics[trim = 50 160 128 150, clip, width = 0.45\textwidth]{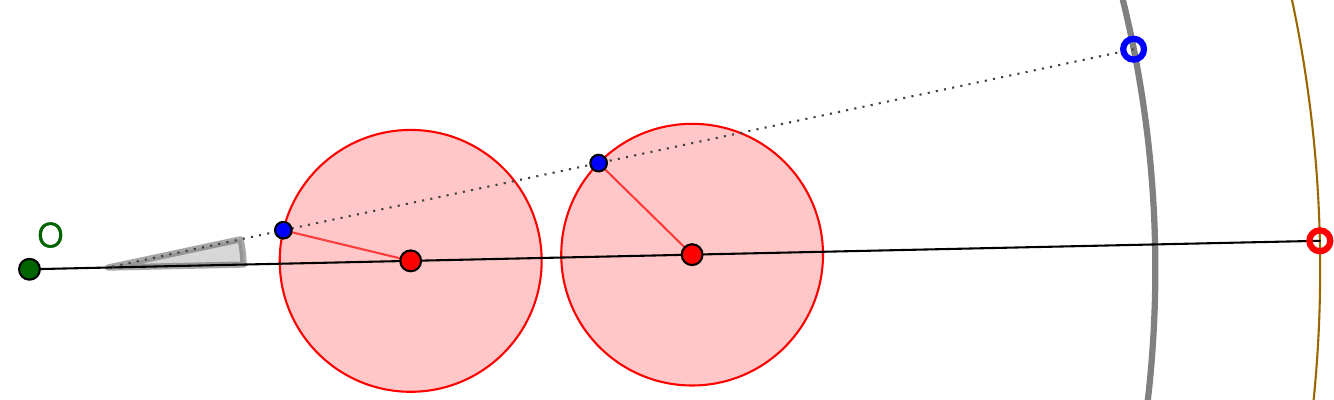}
    \put(-175,15) {$\xa(\taue)$}
    \put(-120,15) {$\xa(\tau)$}
    \put(-115,35) {\re$\ra$}
    \put(-30,25) {$\underset{{\ro+\rt - r}}{\underbrace{~~}}$}
    \put(-195,30) {\color{gray}$\theta$}
    \caption{ {\re $\medbullet$ } and {\color{blue} $\medbullet$ }  denote the location of the intruder and defender at times $\tau$ and $\taue$. The hollow red and blue circles denote their initial locations.
    The brown arc denotes a part of the TSR boundary. 
    The gray arc has a radius of $r$ and it is concentric to the target center (denoted by the green point $O$). 
    The defender moves along the dotted line. 
    The angle between the defender's path and the intruder's path is denoted by $\theta$.
    } \label{fig:sufficient}
\end{figure}
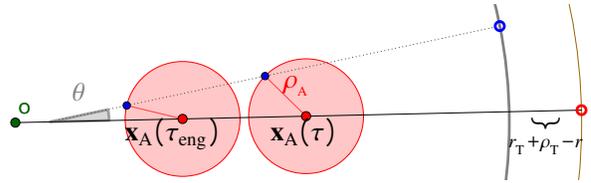
\begin{remark}
 To ensure that $\Se \neq \emptyset$, the parameters must satisfy $2\gamma\ra \le \rt$ which is ensured by \eqref{eq:Asm2} in Assumption \ref{assm:deadlock}.  \\
\end{remark}

\subsection{Capturability of an attacker} \label{Max Angle}
Let the defender start the game from the location $r\uvec({\theta_{\rm D0}})$.
For a given $r$ and $\theta_{\rm D0}$, the \textit{capturability} condition is determined based on the initial angular separation between the defender and the attacker at the time the attacker enters the TSR.
For any arbitrary engagement configuration $(\taue,\thetae)$, the defender can reach this configuration if $|\theta_{\rm A0}\!-\!\theta_{\rm D0}|\le  \theta_{\max}(\taue,\thetae, r)$, where $\theta_{\rm A0}$ is the entry angle of the attacker and 
\begin{subequations} \label{Maxanglesep}
\begin{align} 
    &\theta_{\max}(\taue,\thetae, r) = \cos^{-1} \bigg( \frac{ \reng^2 +r^2 - \taue^2}{2\reng r} \bigg)  + \phie, \\
    &\phie = \sin^{-1}\bigg(\frac{\ra \sin(\thetae)}{\reng} \bigg),  \\
    \begin{split}
    & \reng = \!\bigg((\ro\! + \! \rt-\taue\nu)^2 +\\ &  \ \ \ \  \ \ \  \ \ \! \ra^2 \!+ 2(\ro\! + \!\rt-\taue\nu) \ra\cos(\thetae)\bigg)^{\frac{1}{2}}; 
    \end{split}
    \end{align} 
     \end{subequations}
see \cite{pourghorban2022target} for a detailed derivation.
Therefore, at the beginning of a 1-vs-1 game, if $|\theta_{\rm A0}\!-\!\theta_{\rm D0}|\le  \theta_{\max}(\taue,\thetae, r)$ for some $(\taue, \thetae) \in \Se $ then  capture is guaranteed.
Equation \eqref{Maxanglesep} gives the necessary condition on the defender's initial location (i.e., $r$ and $\theta_{\rm D0}$) to ensure reachability to the engagement surface. 
However, it does not provide any guarantee whether the defender will be detected by the intruder along its way to  $\x_{\rm eng}$, in which case the engagement will start earlier in a different configuration.
For example, a scenario like the one presented in Fig.~\ref{fig:sufficient} may occur where the defender is sensed at $\tau < \taue$. 
In the following theorem we formally prove that a scenario like  Fig.~\ref{fig:sufficient} shall not occur under a certain condition on $r$. 
In other words, the necessary condition in Equation \eqref{Maxanglesep} is also sufficient.
\begin{lemma}[Sufficient]\label{lemma:sufficiency}
  If $r\le \ro + \rt - \ra$, then $\theta_{\rm D0} \le \theta_{\max}(\taue,\thetae,r)$ is a sufficient condition for a defender located at $r\uvec({\theta_{\rm D0}})$ to reach $(\taue, \thetae) \in \Se$ before getting detected by the intruder.
\end{lemma}
\begin{proof}
We will use proof by contradiction for this theorem. To this end, let us assume that the defender is first detected at time $\tau<\taue$ as shown in Fig.~\ref{fig:sufficient}.

Let us first consider the case when $\theta \le \cos^{-1}(\nu)$ in Fig.~\ref{fig:sufficient}, where recall that $v$ denotes the speed ratio.
Therefore, the horizontal velocity of the defender is more than $\nu$.
Since the initial horizontal distance between the defender and the intruder is $\ro+\rt - r \ge \ra$ (due to the theorem's condition), the horizontal distance between them will remain at least $\ra$ as long as the defender keeps moving on the straight line path shown in Fig.~\ref{fig:sufficient}.   
However, in Fig.~\ref{fig:sufficient} we notice that the horizontal distance between them is smaller than $\ra$ at time $\tau$. 
Therefore, the configuration at time $\tau$ is not feasible. 

Now we consider the other case when $\theta > \cos^{-1}(\nu)$.
In this case, the horizontal velocity of the defender is smaller than $\nu$ and thus the horizontal distance between the defender and intruder \textit{decreases} with time. 
However, we notice that the horizontal distance between them has \textit{increased} from time $\tau$ to time $\taue$. 
Therefore, it must be the case that the shown configuration at time $\tau$ does not occur. 
\end{proof}
\subsection{Reachable portion of the Engagement Set} \label{reachsurf}

Let the defender be located at $r\uvec({\theta_{\rm D0}})$ at the time the attacker appears on the TSR boundary at $(\ro+\rt)\uvec(\theta_{\rm A0})$, where $|\theta_{\rm D0} - \theta_{\rm A0}|\le \theta_{\max}(\taue,\thetae, r)$. For a given defender location $\xdo= r\uvec(\theta_{\rm D0})$ and a given attacker location $\xao=(\ro+\rt)\uvec(\theta_{\rm A0})$, the reachable part of the engagement surface is defined as
    \begin{align}
       \lcap(\xdo, \xao) =\{ (\taue,\thetae)\in  \Se ~: 
       \end{align}
       \vspace{-0.7 cm}
       \begin{align*}
       |\theta_{\rm A0}\!-\!\theta_{\rm D0}|\le  \theta_{\max}(\taue,\thetae, r)\}.
    \end{align*}
Therefore, for a defender located at $r\uvec(\theta_{\rm D0})$, only a part of the \textit{Engagement Surface} may be reachable and the defender can choose any point in this \textit{capturable set} $\lcap(\xdo, \xao)$ to engage and start the \textit{Full Information Phase}.
\begin{remark}
    For a given defender location $r\uvec(\theta_{\rm D0})$, an attacker is guaranteed to breach the target if and only if $\lcap(\xdo, \xao)=\emptyset$. In this case, the defender moves radially toward the target center to better position itself for the next 1-vs-1 game. 
\end{remark}
\subsection{Simplified Engagement }\label{sec:simplified}
For a defender, to find the optimal engagement configuration $(\taue, \thetae)$, it performs the max-min optimization in \eqref{eq:stackelberg}. Solving this max-min game is computationally prohibitive in real-time settings. Therefore, in the following, we discuss a simplified version of the solution that retains the same structural framework as the original problem but addresses the computational challenges more effectively.
\par
Given an attacker location $\xao=(\ro+\rt)\uvec(\theta_{\rm A0})$ and a defender location $\xdo= r\uvec(\theta_{\rm D0})$ the simplified reachable part of the engagement surface is defined as
    \begin{align}
       \lscap(\xdo, \xao) =\{ (\taue,\thetae)\in  \lcap(\xdo, \xao) ~: 
       \end{align}
       \vspace{-0.7 cm}
       \begin{align*}
       \| \xc(\taue) \| =\ro+ \gamma \ra \}.
    \end{align*}
Therefore, the defender selects engagement points only from the \textit{capturable set} $\lcap(\xdo, \xao)$, where the corresponding AC is tangent to the target. This ensures that, after each capture, the defender remains closer to the target center than it would by choosing any other engagement location where $| \xc(\taue) | > \ro + \gamma \ra$.

\par
Utilizing the simplified strategy, the defender solves \eqref{eq:guard_condition} only for the equality case for each $\taue$ and computes the corresponding pair $(\taue,\thetae)$. Consequently, $\Se$ in \eqref{eq:SE} is restricted to those $(\taue,\thetae)$ pairs that satisfy \eqref{eq:critical_theta} with equality, thereby significantly reducing the number of candidate solutions the defender must evaluate when solving \eqref{eq:stackelberg}.
Under Assumption~\ref{assm:deadlock}, this strategy also guarantees that, in the event of an attacker breach, the defender can return to the target center before the next game begins. However, there may exist $(\taue,\thetae) \in \lcap(\xdo, \xao)$ such that $(\taue,\thetae) \notin \lscap(\xdo, \xao)$. In such cases, by selecting from the simplified set, the defender may lose the current game—even though capture would have been possible under the optimal strategy.
\par
Since these optimal engagement points lie within a well-defined geometric region, the corresponding capture locations delineate a particularly structured area of interest, which we refer to as the \textit{capture cone}.

\begin{remark}\label{remark:capturecone}
All captures occur within the capture cone, which is concentric with the target, shares its angular span of $|\Phi|$, and has a radius of $\ro + 2\gamma \ra$.
\end{remark}
\section{Capture Probability} \label{sec:capprob}
In this section, we discuss the computation of $P(\x)$. If the defender is at the target center, then, due to Assumption~\ref{assm:small target}, capture is guaranteed with probability $1$ regardless of the attacker's arrival angle $\ta$.
 Otherwise, the capture probability depends on the defender's starting location $\x$.

\begin{figure}
    \centering
    \includegraphics[trim = 110 50 380 65, clip, width=0.4\linewidth]{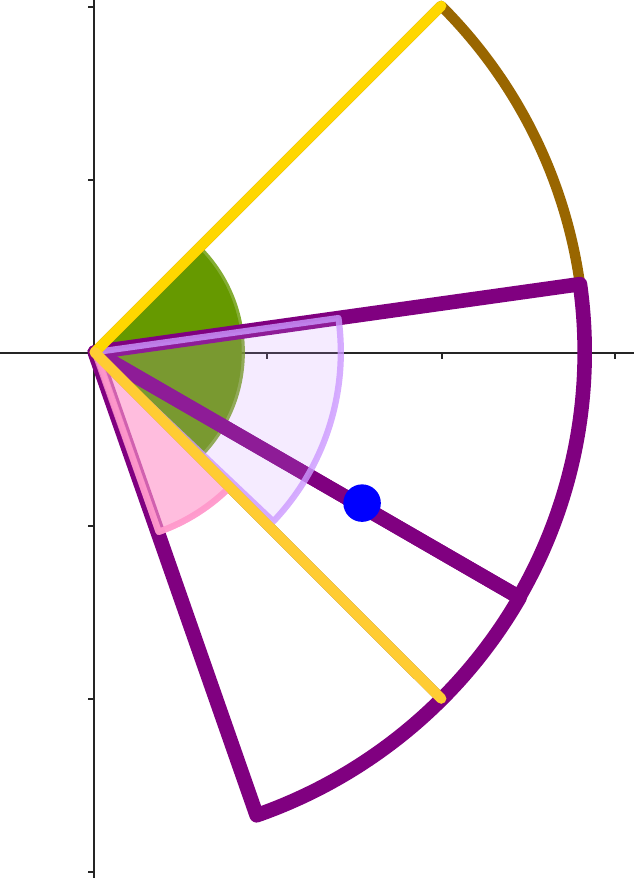}
    \put (-52, 56 ) {$\theta_{\max}(r)$}
    \put (-120, 30 ) {$|\theta_{\rm D0}|+\theta_{\max}(r)-\Phi$}
    \hspace{0.2 cm}
       \includegraphics[ trim = 215 50 670 10, clip,width=0.35\linewidth]{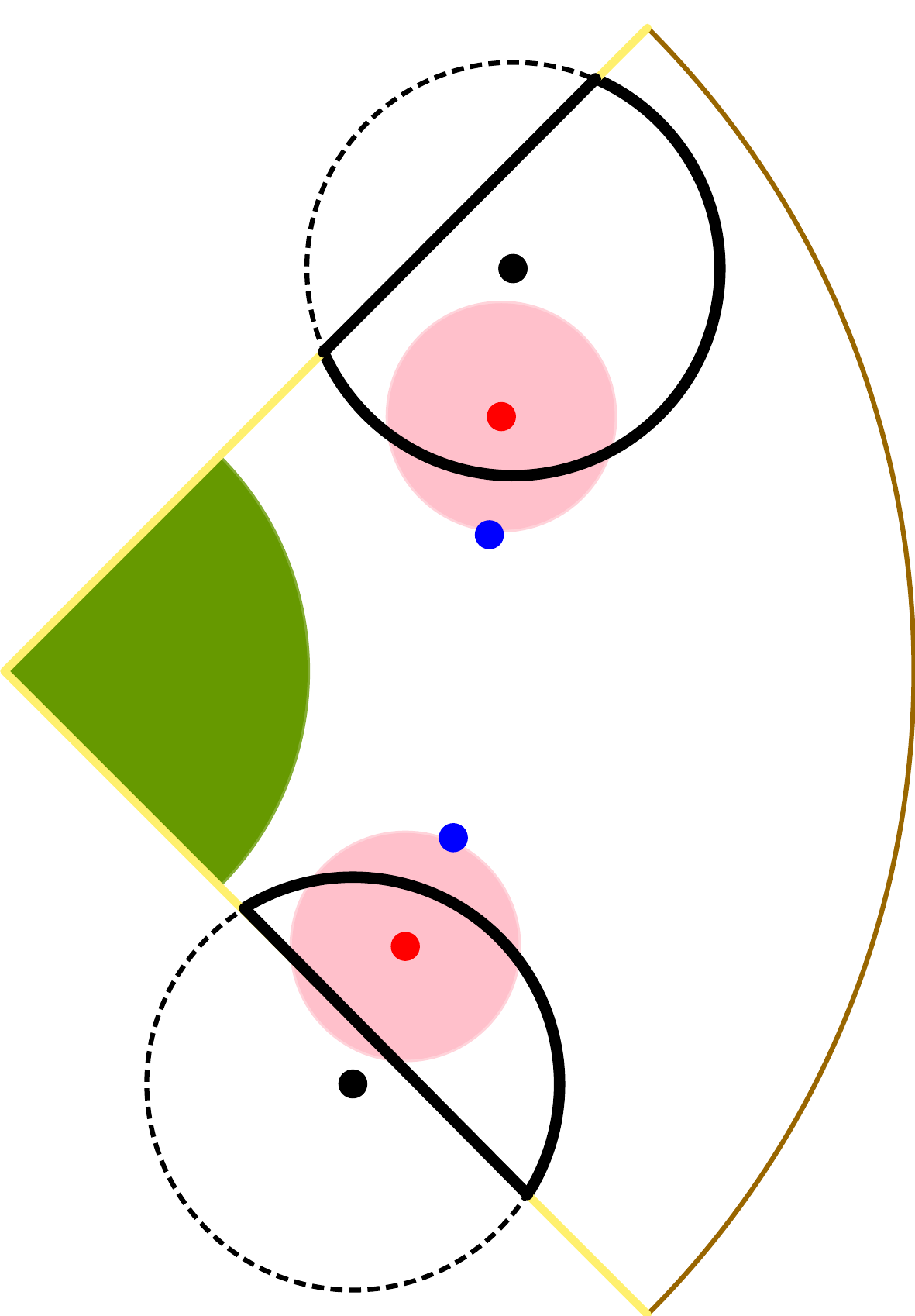}
    \caption{Left: For a defender located at $r\uvec({\theta_{\rm D0}})$, the configuration when $|\theta_{\rm D0}| > \Phi-\theta_{\max}(r)$. Right: The configuration of the Apollonius circle for two different engagements. In the top engagement, the Apollonius circle has an intersection with the edges of the environment.}
    \label{fig:specialcase}
\end{figure}
 \begin{figure}
\label{capture_distribution}
    \centering
    \includegraphics[trim = 215 250 120 250, clip, width=0.4\linewidth]{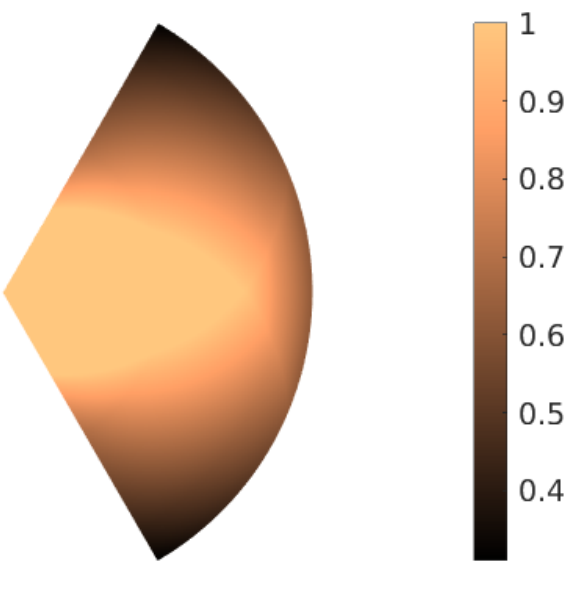}
       \includegraphics[ trim = 140 230 75 225, clip,width=0.5\linewidth]{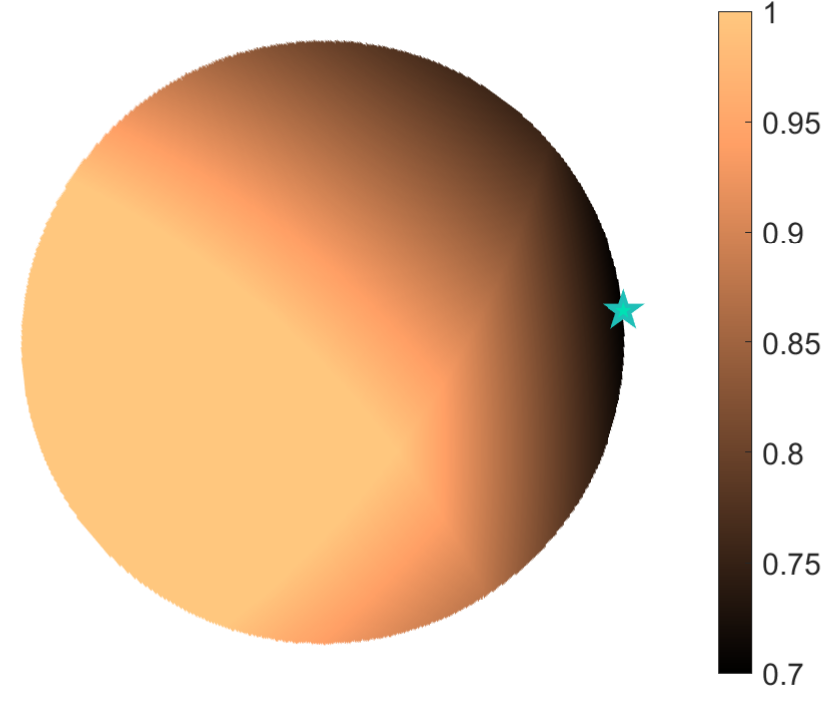}
    \caption{Left: Capture distribution in the game environment.
Right: Capture distribution within an arbitrary AC in the conical environment, with the corresponding capture point marked by a light blue star. }
    \label{fig:capture_distribution}
\end{figure}
 
 \begin{lemma} \label{capture_prob}
Let \( \x = r\uvec(\theta_{\rm D0}) \) denote the defender's location.  
  For a given $r$, by maximizing $\theta_{\max}(\taue,\thetae, r)$ w.r.t $(\taue,\thetae)$ with the constraints \eqref{eq:critical_theta} and \eqref{eq:critical_theta_guard}, we obtain the pair $(\taue^*, \thetae^*)$. For notational brevity, we write \( \theta_{\max}(r) \triangleq \theta_{\max}(\taue^*, \thetae^*, r) \). 
Then, the capture probability is given by
\begin{equation}
\hspace{-1mm}
\begin{aligned}
P(\x) =
\begin{cases}
\frac{\theta_{\max}(r)}{\Phi},
&\!\!\! \text{if } |\theta_{\rm D0}| \le \Phi - \theta_{\max}(r), \\
\frac{\theta_{\max}(r) + \Phi - \theta_{\rm D0}}{2\Phi},
&\!\!\! \text{if } |\theta_{\rm D0}| > \Phi - \theta_{\max}(r).
\end{cases}
\end{aligned}
\end{equation}
\end{lemma}

\vspace{0.4 cm}
 \begin{proof}
    For a defender located at $\x=r\uvec({\theta_{\rm D0}})$, given that the attackers appear on the TSR boundary independently with uniform probability, if $|\theta_{\rm D0}| \le \Phi-\theta_{\max}(r)$ the capture probability is
    \begin{align*}
        P(\x)= \frac{\theta_{\max}(r)}{\Phi}.
    \end{align*}
    However, if a configuration such as the one in the left subplot of Fig.~\ref{fig:specialcase} occurs where $|\theta_{\rm D0}| > \Phi-\theta_{\max}(r)$, the defender is unable to utilize a part of its $\theta_{\max}$, and in this case the capture probability is
    \begin{align*}
       P(\x)= \frac{\theta_{\max}(r)+\Phi-\theta_{\rm D0}}{2\Phi}.
    \end{align*}
    This completes the proof.
 \end{proof}
 
 Therefore, for a defender located at any point inside the TSR, there is a probability of capturing for the next incoming attacker. We call the collection of these probabilities as \textit{Capture Distribution} which is shown with the cooper color bar in Fig.~\ref{fig:capture_distribution}.
 \par 
 Having derived an analytical expression for the capture probability \( P(\x) \), we are now equipped to solve the max-min game formulation in~\eqref{eq:stackelberg}, where the defender selects an engagement configuration to optimize its long-term performance.

\section{Game Progression} \label{sec:GameProgress}

The first game begins with the defender located at the target center. Due to Assumption~\eqref{assm:small target}, the defender captures the first incoming attacker with probability one. After this capture, let the defender be positioned at $r_1\uvec(\theta_{\rm D_1})$, and the probability of capturing the next attacker becomes $p_1 = P(r_1\uvec(\theta_{\rm D_1}))$.

Note that $r_1$ and $\theta_{\rm D_1}$ depend on the angle of arrival of the first attacker and are therefore random variables. When the next attacker appears on the TSR boundary, it is considered capturable if $|\theta_{\rm D_1} - \theta_{\rm A}| \le \theta_{\max}(r_1)$. Let the capture location for this attacker be $r_2\uvec(\theta_{\rm D_2})$, and therefore the defender is able to capture the next attacker with probability $p_2 = P(r_2\uvec(\theta_{\rm D_2}))$.

On the other hand, if $|\theta_{\rm D_1} - \theta_{\rm A}| > \theta_{\max}(r_1)$, the defender has no strategy to capture the second attacker and instead moves back toward the target center. This is because, under Assumption~\eqref{assm:small target}, the defender can always capture an attacker when starting from the target center.

Each time the defender fails to capture an attacker and returns to the target center to guarantee capture of the next one, we refer to this as a \textit{restart}. The defender’s strategy is described by the following algorithm.

\begin{algorithm} \label{algo}
    \caption{Defender's Strategy}
    \label{euclid}
    \begin{algorithmic}[1] 
    \State Initialize $\xd \gets [0,0]^\intercal$, $N_{\rm capture}\gets 0$,  and $N$
    \For{$n = 1: N$}
    \State $\theta_{\rm A0} \sim {\mathcal U}(-\Phi, \Phi)$ \Comment{Uniform random arrival of attacker}
    \If{$\xd$ is at target center} \Comment{{\color{blue!60}Capture happens}}
       \State $N_{\rm capture} \gets N_{\rm capture} +1$
       \State $\xd \gets \xcap$ \ \ \ \ \ \ \  from \eqref{eq:capture_point} \Comment{$\xd$ after capture}
    \Else
    \If{$|\theta_{\rm D0} - \theta_{\rm A0}|\le \theta_{\max}(r)$} \Comment{{\color{blue!60}Capture happens}}
        \State $N_{\rm capture} \gets N_{\rm capture} +1$
        \State $\xd \gets \xcap$ \ \ \ \  from \eqref{eq:capture_point} \Comment{$\xd$ after capture}
       \Else \Comment{{\color{red!60} Breach happens}}
       \State Defender goes to the target center (restart)
       \State $\xd \gets [0,0]^\intercal$ 
       \EndIf
    \EndIf
    \EndFor
    \end{algorithmic}
\end{algorithm}
\subsection{Non-Convex Conic Problem}
We solve the game for a convex conic environment with angular span $\Phi \le \frac{\pi}{2}$. However, one can also consider the game in a non-convex conic environment where $\tfrac{\pi}{2} <\Phi < \pi$. The solution presented in the following sections extends to this more general setting as well. Specifically, when $|\theta_{\rm D0} - \theta_{\rm A0}| \le \pi$, the defender plays the game exactly as described in the subsequent sections. However, when $|\theta_{\rm D0} - \theta_{\rm A0}| > \pi$, the defender first moves to the target center and resumes the game from there, as illustrated in Fig.~\ref{fig:nonconvex_cone}. A detailed analysis is beyond the scope of this work and is a potential future work.
\begin{figure}
    \centering
    \includegraphics[trim = 40 40 20 20, clip,width = 0.25 \textwidth]{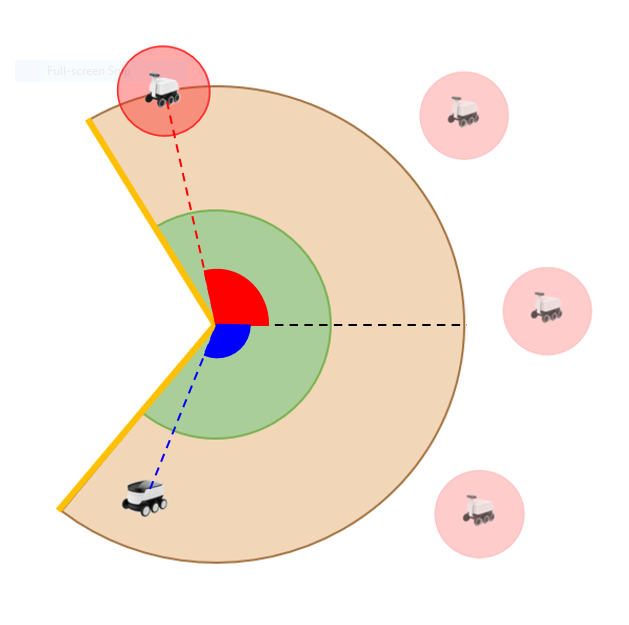}
     \put (-99.2, 23.5) {\line(1,2.5) {12.8}}
  \put (-94, 90 ) {\line(1,-4.5) {7.7}}
   \put(-83,42) {$\color{blue}{\theta_{\rm D0}}$}
  \put(-76,65) {$\re{\theta_{\rm A0}}$}
           \caption{Game progression in a non-convex conic environment when $|\theta_{\rm D0} - \theta_{\rm A0}| > \pi$.}
    \label{fig:nonconvex_cone}
    \vspace{-12 pt}
\end{figure}
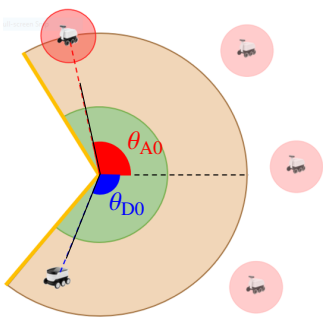
\section{Bounds on Capture Percentage} \label{sec:bound}
\begin{figure}  
    \centering
    \includegraphics[trim = 250 140 50 130, clip, width=0.4\linewidth]{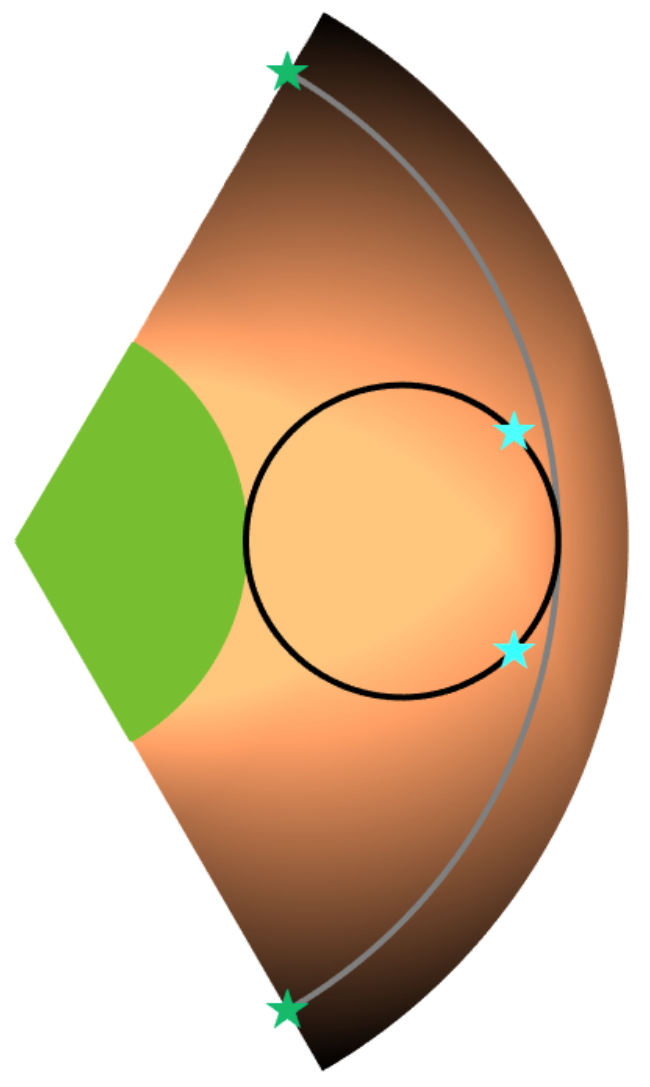}
           \caption{Capture Points. The light blue stars are the capture locations with the highest probability $\x^*$ of capturing the next incoming attacker. Green stars are the points with the lowest probability inside the capture cone.}\label{fig:capture_points}
\end{figure}
We measure the performance of the defender by the capture percentage, i.e., the ratio of the number of attackers captured and the number of attackers who arrived multiplied by $100$. 
In this section, we compute a lower and an upper bound on the capture percentage.
\par
The defender captures the first incoming attacker with probability one. Based on its location after this capture, the defender either captures the second attacker with probability $p_1$, or returns to the target center with probability $1 - p_1$.
If the defender successfully captures the second attacker, it then has probability $p_2$ of capturing the third incoming attacker, or it fails with probability $1 - p_2$ and returns to the target center. This sequence of events resembles a Markov chain over a continuous state space.
However, due to the randomness in the defender’s capture probability after each attacker's arrival, the overall capture percentage becomes analytically and computationally intractable. To address this, we derive analytical lower and upper bounds on the capture percentage by employing corresponding lower and upper bounds on the defender’s capture probability.

\par
 Recall from Remark.~\ref{remark:capturecone}, the capture happens within the \textit{capture cone}.
 However, the capture location depends on the initial location of the defender and the attacker. Therefore, the location of the defender at the end of each attacker's arrival is a random point inside the \textit{capture cone} and subsequently, the corresponding capture probability of the defender is a random variable. 
   \par
   A lower bound on the capture probability is $p^*=P(\x=(\ro+2\gamma\ra)\uvec(\pm \Phi)) $ when the defender is on the boundaries of the \textit{capture cone} which are shown with green stars in
   Fig.~\ref{fig:capture_points}.

 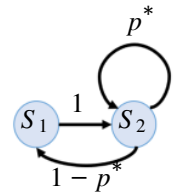
\begin{figure}
     \centering
    { \includegraphics[ angle=90,width=0.13\textwidth]{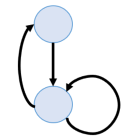}}
     \put(-59, 23) {$S_1$}
      \put(-22, 23) {$S_2$}
      \put(-19, 60) {$p^*$}
      \put(-48, 1) {$1-p^*$}
       \put(-40, 29) {$1$}
      \caption{Lower-bound on capture percentage Markov chain.}
      \label{markovfig}
    \end{figure}
\begin{lemma} \label{lemma:lowerbound_percentage}
    A lower bound on the capture percentage after $N$ attacker's arrival is
    \begin{align}
                {\rm Percentage}(N)\ge \frac{\sum_{i=1}^N \p^\intercal \eta_i}{N}\times 100,
                \label{eq:lowmarkovpercentage}
    \end{align}
        where $\eta_i = [\eta_i(1), \eta_i(2)]^\intercal$ with $\eta_i(j)$ denote the probability that the sate of the Markov chain is at $S_j$ at the end of the $i$-th 1-vs-1 game and $\p=[1,  p^*]$.
\end{lemma}
\begin{proof}
    Let $w_n$ denote the total number of captures by the end of the $n$-th game. Let us further define the random variable $s_n \in \{S_1, S_2\}$ to denote the state of the Markov chain  at time $n$; see Fig.~\ref{markovfig}. 
    Let us define the binary random variables $\mu_1$ and $\mu_2$ such that 
    \begin{align*}
        \mu_i = \begin{cases}
            1, \qquad &\text{Attacker is captured from the game state } S_i,\\
            0, &\text{Attacker is not captured }
        \end{cases}
    \end{align*}
    Therefore, we may write
    \begin{align}
    \label{eq:w_n}
          w_{n+1} = w_n + \mu_1 \mathds{1}_{S_1}(s_n) + \mu_2 \mathds{1}_{S_2}(s_n), 
    \end{align}
    where $\mathds{1}_s(\cdot)$ is an indicator function such that $\mathds{1}_s(s') = 1$ if and only if $s=s'$; otherwise, $\mathds{1}_s(s') = 0$.
    Furthermore, since the attackers appear independently of the game state state, $\mu_i$ and $s_n$ are independent random variables. 
    Taking expectations on both sides of \eqref{eq:w_n}, we obtain
    \begin{align} \label{eq:expected_w_n}
    \varmathbb{E}[w_{N+1}] = \varmathbb{E}[w_N] + \eta_{N}(1) +  p^* \eta_N(2),
\end{align}
where $\eta_n(j)$ denotes the probability that the sate of the Markov chain $(s_n)$ is at $S_j$ at the end of the $n$-th game. By defining the vectors $\p=[1, p^*]^\intercal$ and $\eta_n=[\eta_N(1),\eta_N(2)]^\intercal$, we may rewrite \eqref{eq:expected_w_n} as
\begin{align}
    \varmathbb{E}[w_{n+1}] = \varmathbb{E}[w_n] + \p^\intercal \eta_n=\varmathbb{E}[w_0]+\sum_{i=1}^n \p^\intercal \eta_i = \!\sum_{i=1}^n \p^\intercal \eta_i,
\end{align}
where $w_0=0$. Therefore, the expected percentage of capture at the end of the $n$-th game is,
    \begin{align*}
        {\rm percentage}(n)=\frac{\varmathbb{E}[c_n]}{n}\times 100= \frac{\sum_{i=1}^n \p^\intercal \eta_i}{n}\times 100.
    \end{align*}
    This completes the proof.
 \end{proof}
An upper bound on the capture probability is given by \( q^* = P(\x^*) \), where the corresponding AC is located in the middle of the TSR, i.e., \( \xc(\taue) = (\ro + \gamma\ra)\uvec(0) \), and \( \x^* \) is the solution to~\eqref{eq:capture_point}. These solutions are depicted as light blue stars in Fig.~\ref{fig:capture_points}.

\begin{lemma}
    An upper bound on the capture percentage after $N$ attacker's arrival is
    \begin{align}
                {\rm Percentage}(N)\le \frac{\sum_{i=1}^N \q^\intercal \eta_i}{N}\times 100,
                \label{eq:highmarkovpercentage}
    \end{align}
        where $\eta_i = [\eta_i(1), \eta_i(2)]^\intercal$ with $\eta_i(j)$ denote the probability that the sate of the Markov chain is at $S_j$ at the end of the $i$-th 1-vs-1 game and $\q=[1,  q^*]$.
\end{lemma}
\begin{proof}
    The proof of this lemma is similar to Lemma~\ref{lemma:lowerbound_percentage}.
\end{proof}


\section{Simulation Results} \label{sec:Simu}
\subsection{Monte-Carlo Simulation}
We simulate the game with the following parameters $\ro =6$, $\ra = 1,~ \rt = 8$,  $\nu = 0.85$, and $\Phi=\frac{\pi}{3} $. 
Under this parametric choice, we conducted $100$ random trials of the game.
In each trial, we considered a sequence of $200$ incoming attackers that are uniform randomly generated on the TSR boundary. The percentage of capture for each trial is plotted in Fig.~\ref{fig:capture_percentage}. 
\begin{figure}
    \centering
        \includegraphics[trim = 85 245 110 220, clip, width=0.75\linewidth]{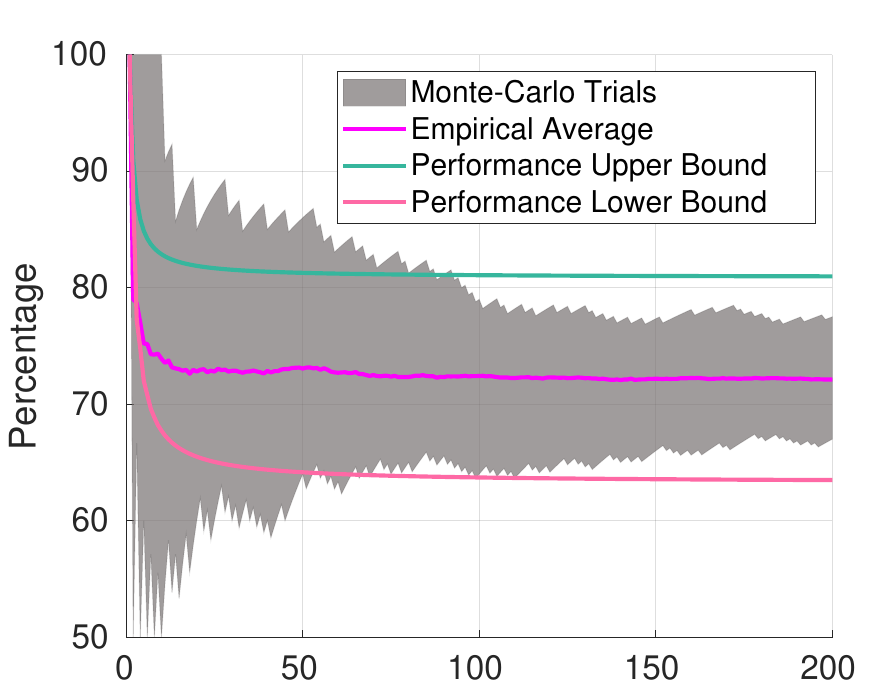}
           \caption{Percentage of Capture versus number of arrivals from 100 trials. Each trial is represented with a gray line. 
    Their empirical average is represented by the magenta line.
    The pink and teal lines represent the theoretically predicted lower/upper-bounds on the capture percentage in \eqref{eq:lowmarkovpercentage} and \eqref{eq:highmarkovpercentage}.}
    \label{fig:capture_percentage}
\end{figure}
The abscissa in this figure represents the number of arrivals ($N$), while the ordinate represents the percentage of captures, ${\rm Percentage}(N)$, corresponding to that number of arrivals. We compute the empirical mean of the capture percentage  from these random trials, and that is shown by the magenta line in Fig.~\ref{fig:capture_percentage}.
To compare this simulation result with our theoretical bounds, we plot the lower and the upper bound on the capture percentage in the same figure using pink and teal lines, respectively. We observe that the empirical average from the trials converges after a few sequences of attacker arrivals and by the end of the $200$ sequences the capture percentage is equal to $72.11$. The theoretical lower and upper bounds from \eqref{eq:lowmarkovpercentage} and \eqref{eq:highmarkovpercentage}, with the chosen parameters are $63.48$ and $80.96$, respectively. The empirical average remains within the theoretical bounds, validating our analysis. The narrow bound gap shows the model’s accuracy, and the rapid convergence suggests that performance can be predicted with few observations, enabling timely decisions.
\begin{figure}
    \centering
    \includegraphics[trim = 85 220 110 220, clip, width=0.75\linewidth]{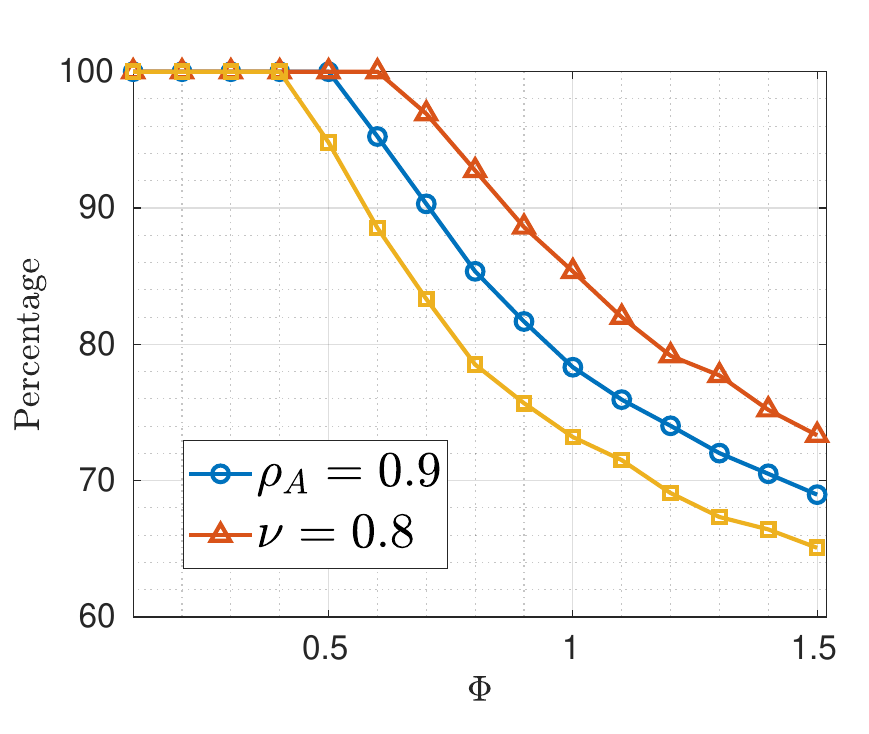}
    \caption{The capture percentage is plotted against $\Phi$. The yellow curve  represents the following parameters: $\ro = 6$, $\ra = 1$, $\rt = 8$, and $\nu = 0.85$. The capture percentage is plotted with blue and red lines when $\ra=0.9$, and $\nu=0.8$, respectively.}
    \label{fig:angle_variation}
\end{figure}
\subsection{Parameter Variation}
In the next simulation, we vary $\Phi$ and report the empirical mean of the capture percentage. The abscissa in Fig.~\ref{fig:angle_variation} denotes the conical environment angle ($\Phi$) and the ordinate denotes the percentage of capture (i.e., ${\rm Percentage}$). For each angle $\Phi$ we again conducted $100$ trials of the game where in each trial a sequence of $200$ incoming attackers were tasked to attack the target. First, we simulate the game using the following parameters: $\ro =6$, $\ra = 1,~ \rt = 8$, and $\nu = 0.85$.  For each value of $\Phi$, the empirical average of the capture percentage from the conducted trials is reported, and they are all shown with a yellow line.  In the next two simulations, we keep the environment parameters (i.e., $\ro=6$, $\rt=8$) fixed and will vary the parameters associated with the attackers' capability (i.e., $\ra$, $\nu$). First, we fix $\nu=0.85$ and consider $\ra=0.9$ and then we fix $\ra=1$ and consider $\nu=0.8$. The empirical average of the capture percentage from the same experiment setup is indicated by blue and red lines when $\ra=0.9$ and $\nu=0.8$, respectively.
It is evident that the capture percentage decreases as $\Phi$ increases. This occurs because the defender's probability of capturing the next incoming attacker decreases with higher values of $\Phi$.
When $\Phi$ is small, the initial angular separation between the defender and attacker after each capture is more likely to be less than $\theta_{\max}$ and the defender can capture almost all the incoming attackers. 
 Decreasing $\ra$ and $\nu$ has increased the capture percentage which was expected from \eqref{Maxanglesep}. However, from the simulation, the defender captured more attackers by varying $\nu$ for any $\Phi$.
\subsection{Algorithm Comparison}

In this set of simulations, we evaluate the performance of our proposed approach in comparison with the state-of-the-art Pure-Pursuit strategy. 

We conduct four experiments, each corresponding to a different combination of defender and attacker strategies:
\begin{itemize}
    \item Our Defender – Our Attacker
    \item PP (Pure Pursuer, i.e., always moves directly toward the attacker’s current position) Defender – Our Attacker
    \item PP Defender – PE (Pure Evader, i.e., moves toward the target center until senses the defender and moves directly away from the opponent’s current position) Attacker
    \item Our Defender – PE Attacker
\end{itemize}

In each experiment, we run 100 independent trials of the game. In every trial, the defender faces a sequence of 200 incoming attackers attempting to reach the target. The percentage of capture for each trial is plotted in Fig.~\ref{fig:algo_compare}. Similar to Fig.~\ref{fig:capture_percentage},

From the simulation results, it is evident that the PP defender struggles in both the experiments, showing significantly lower capture rates. In particular, when the PP defender faces the PE attacker, its performance drops sharply and converges at a very low success rate—indicating poor adaptability to evasive behaviors.

On the other hand, our proposed defender strategy consistently outperforms the PP defender. When paired with our attacker, the \textit{Our Defender – Our Attacker} configuration maintains a high and stable success rate even as the number of attackers increases. The \textit{Our Defender – PE Attacker} configuration also performs robustly, showing that our approach can effectively counter even highly evasive attackers.

These findings clearly illustrate that the Pure Pursuit approach is insufficient for dealing with a large number of dynamic and potentially evasive attackers. In contrast, our strategy provides significantly higher capture rates and better long-term defensive performance. These results indicate that the pairing of our defender and attacker strategies exhibits equilibrium-like behavior: any unilateral deviation from the prescribed strategy leads to a decrease in the deviating agent’s performance.

 \begin{figure}
    \centering
    \includegraphics[trim = 80 225 110 230, clip, width=0.75\linewidth]{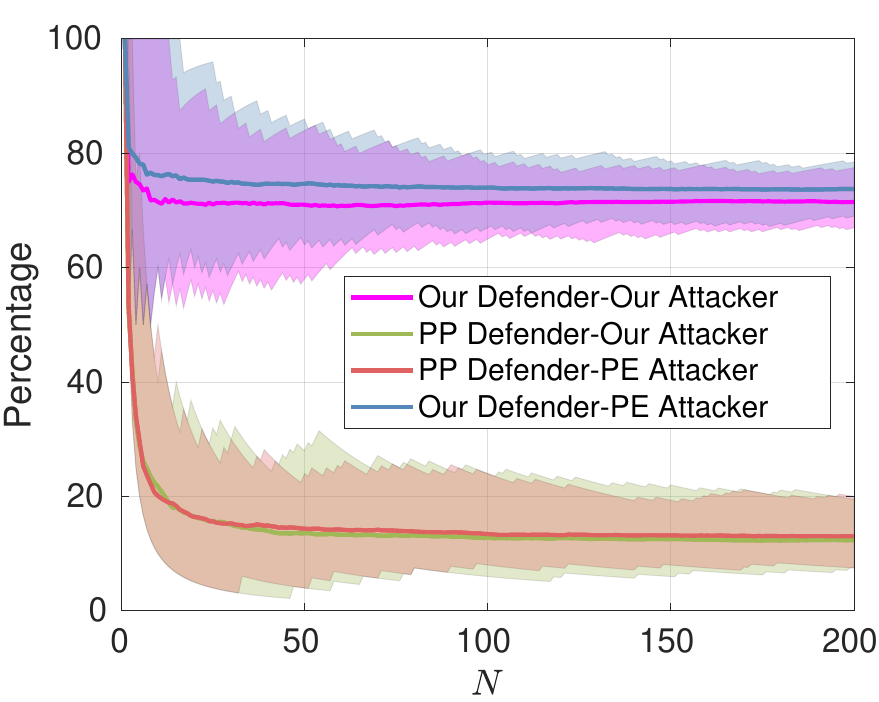}
           \caption{Percentage of Capture versus number of arrivals from 100 trials. }
    \label{fig:algo_compare}
\end{figure}
\section{Conclusions} \label{sec:Conclusion}
In this paper, we formulated a perimeter defense game against a sequence of incoming attackers in a  conical environment. The edges of the conical environment are protected by impenetrable obstacles and the attackers are tasked to move radially toward the target center the circumference of the conic environment and breach the target while the defenders are tasked to capture as many
attackers as possible. Based on the initial configuration of the agents and the information
available to them, we discussed the agents' strategies using the notions of \textit{capture distribution}, and \textit{min-max strategy}. We analytically computed bounds on the performance of the defender. The simulation results closely align with the derived theoretical bounds, with the empirical capture rates consistently falling between the lower and upper limits. 
\par
A natural extension of this work would be to consider different attacker arrival patterns (e.g., periodical arrivals, non-uniform probability of arrival locations, multiple simultaneous arrivals). Another extension of this work would be to consider guarding moving targets and the problem where the defenders are also equipped with their own sensing ranges.





\begin{IEEEbiography}
[{\includegraphics[width=1in,height=1.25in,clip,keepaspectratio]{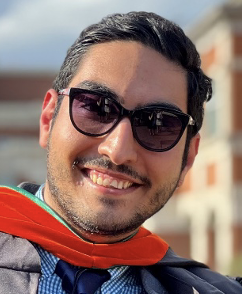}}]
{Arman Pourghorban}
received the B.E. degree in electrical engineering from Isfahan
University of Technology, Iran, in 2020. He is currently a Ph.D. candidate with the Department
of Electrical and Computer Engineering, University of
North Carolina at Charlotte, Charlotte, NC, USA, and also works as a Control Engineer at Tesla, Inc.
His research interests include cooperative control, multi-agent systems, and differential
games.
\end{IEEEbiography}

\begin{IEEEbiography}
[{\includegraphics[width=1in,height=1.25in,clip,keepaspectratio]{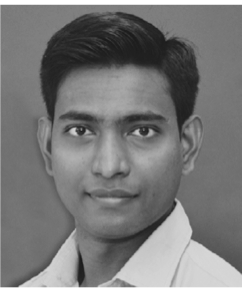}}]
{Dipankar Maity}
(Senior Member, IEEE), received the B.E. degree in electronics and telecommunication engineering from Jadavpur
University, India, in 2013, and the Ph.D. degree in electrical and computer engineering from the University of  
Maryland, College Park, MD, USA in 2018.
He is an Assistant Professor with the Department
of Electrical and Computer Engineering, University of
North Carolina at Charlotte, Charlotte, NC, USA. He
was a Postdoctoral Fellow with the Georgia Institute
of Technology, Atlanta, GA, USA. During his Ph.D., he
was a Visiting Scholar at the Technical University of
Munich (TUM) and at the KTH Royal Institute of Technology, Stockholm, Sweden.
His research interests include temporal logic-based controller synthesis, control
under communication constraints, intermittent-feedback control, multi-agent
systems, stochastic games, and integration of these ideas in the context of
cyber–physical systems.
\end{IEEEbiography}


\begin{thebibliography}{[34]}
\setcounter{enumiv}{0}
\bibitem{lewin1979conic}Lewin, J and Olsder, GJ
\newblock Conic surveillance evasion
\newblock \emph{Journal of Optimization Theory and Applications}, vol. 27, pp. 107--125, 1979.
\bibitem{melikyan1991simple}Melikyan, AA and Ovakimyan, NV
\newblock A simple pursuit-and-evasion game on a two-dimensional cone
\newblock \emph{Journal of Applied Mathematics and Mechanics}, vol. 55, pp. 607--618, 1991.

\bibitem{bajaj2024multi}Bajaj, S., Bopardikar, S. D., Torng, E., Von Moll, A., and Casbeer, D. W.
\newblock Multi-vehicle perimeter defense in conical environments
\newblock \emph{IEEE Transactions on Robotics}, vol. 40, pp. 1439-1456, 2024.
\bibitem{lee2022vision}Lee, E. S., Loianno, G., Jayaraman, D., and Kumar, V.
\newblock Vision-based perimeter defense via multiview pose estimation
\newblock \emph{arXiv preprint arXiv:2209.12136}, 2022.
\bibitem{ho2022game}Ho, E., Rajagopalan, A., Skvortsov, A., Arulampalam, S., and Piraveenan, M.
\newblock Game theory in defence applications: A review
\newblock \emph{Sensors}, vol. 22, no. 3, pp. 1032, 2022.
\bibitem{garcia2020barrier}Garcia, E., Casbeer, D. W., and Pachter, M.
\newblock The barrier surface in the cooperative football differential game
\newblock \emph{arXiv preprint arXiv:2006.03682}, 2020.

\bibitem{eloycoop}Garcia, E., Casbeer, D. W., Pham, K., and Pachter, M.
\newblock Cooperative aircraft defense from an attacking missile
\newblock \emph{Proceedings of the 53rd Conference on Decision and Control}, pp. 2926--2931, 2014.
\bibitem{s2020airport}Lykou, G., Moustakas, D., and Gritzalis, D.
\newblock Defending airports from UAS: A survey on cyber-attacks and counter-drone sensing technologies
\newblock \emph{Sensors}, vol. 20, no. 12, pp. 3537, 2020.
\bibitem{hoorfar2023securing}Hoorfar, H., Fathi, F., Largani, S. M., and Bagheri, A.
\newblock Securing pathways with orthogonal robots
\newblock \emph{arXiv preprint arXiv:2308.10093}, 2023.
\bibitem{isaacs1999differential}Isaacs, R.
\newblock Differential games I, II, III, IV
\newblock RAND Corporation Research Memorandum RM-1391, 1399, 1411, 1468, 1954--1956.


\bibitem{velhal2022decentralized}Velhal, S., Sundaram, S., and Sundararajan, N.
\newblock A decentralized multirobot spatiotemporal multitask assignment approach for perimeter defense
\newblock \emph{IEEE Transactions on Robotics}, vol. 38, no. 5, pp. 3085--3096, 2022.

\bibitem{lee2021guarding}Lee, Y., and Bakolas, E.
\newblock Guarding a convex target set from an attacker in Euclidean spaces
\newblock \emph{IEEE Control Systems Letters}, vol. 6, pp. 1706--1711, 2021.

\bibitem{chen2016multiplayer}Chen, M., Zhou, Z., and Tomlin, C. J.
\newblock Multiplayer reach-avoid games via pairwise outcomes
\newblock \emph{IEEE Transactions on Automatic Control}, vol. 62, no. 3, pp. 1451--1457, 2016.

\bibitem{garcia2019strategies}Garcia, E., Von Moll, A., Casbeer, D. W., and Pachter, M.
\newblock Strategies for defending a coastline against multiple attackers
\newblock \emph{Proceedings of the 58th IEEE Conference on Decision and Control}, pp. 7319--7324, 2019.
\bibitem{chen2014path}Chen, M., Zhou, Z., and Tomlin, C. J.
\newblock A path defense approach to the multiplayer reach-avoid game
\newblock \emph{Proceedings of the 53rd IEEE Conference on Decision and Control}, pp. 2420--2426, 2014.

\bibitem{festa2013decomposition}Festa, A., and Vinter, R. B.
\newblock A decomposition technique for pursuit-evasion games with many pursuers
\newblock \emph{Proceedings of the 52nd IEEE Conference on Decision and Control}, pp. 5797--5802, 2013.

 \bibitem{selvakumar2016evasion}Selvakumar, J., and Bakolas, E.
 \newblock Evasion from a group of pursuers with a prescribed target set for the evader
 \newblock \emph{Proceedings of the American Control Conference}, pp. 155--160, 2016.
 

\bibitem{sun2017pursuit}Sun, W., Tsiotras, P., Lolla, T., Subramani, D. N., and Lermusiaux, P. F. J.
\newblock Pursuit-evasion games in dynamic flow fields via reachability set analysis
\newblock \emph{Proceedings of the American Control Conference}, pp. 4595--4600, 2017.
\bibitem{liang2023collaborative}Liang, X., Zhou, B., Jiang, L., Meng, G., and Xiu, Y.
\newblock Collaborative pursuit-evasion game of multi-UAVs based on Apollonius circle in the environment with obstacle
\newblock \emph{Connection Science}, vol. 35, no. 1, pp. 2168253, 2023.

\bibitem{dorothy2021one}Dorothy, M., Maity, D., Shishika, D., and Von Moll, A.
\newblock One Apollonius circle is enough for many pursuit-evasion games
\newblock \emph{Automatica}, vol. 163, pp. 111587, 2024.


\bibitem{khrenov2021geometric}Khrenov, M., Rivera-Ortiz, P., and Diaz-Mercado, Y.
\newblock Geometric feasibility for defense manifold maintenance in planar reach-avoid games against a fast evader
\newblock \emph{IFAC-PapersOnLine}, vol. 54, no. 20, pp. 807--813, 2021.

\bibitem{deng2023multiple}Deng, R., Zhang, W., Yan, R., Shi, Z., and Zhong, Y.
\newblock Multiple-pursuer single-evader reach-avoid games in constant flow fields
\newblock \emph{IEEE Transactions on Automatic Control},  vol. 69, no. 3, pp.1789-1795, 2023.
\bibitem{wang2024target}Wang, K., Zhou, S., Yao, Y., Sun, Q., and Wang, Y.
\newblock A target defence-intrusion game with considering the obstructive effect of target
\newblock \emph{IET Control Theory \& Applications}, vol. 18, no. 18, pp.2660-2674, 2024.

\bibitem{bajaj2022competitive}Bajaj, S., Torng, E., Bopardikar, S. D., Von Moll, A., Weintraub, I., Garcia, E., and Casbeer, D. W.
\newblock Competitive perimeter defense of conical environments
\newblock \emph{Proceedings of the 61st Conference on Decision and Control}, pp. 6586--6593, 2022.

\bibitem{adler2022rolenew}Adler, A., Mickelin, O., Ramachandran, R. K., Sukhatme, G. S., and Karaman, S.
\newblock The role of heterogeneity in autonomous perimeter defense problems
\newblock \emph{Algorithmic Foundations of Robotics XV}, pp. 115--131, 2023.

\bibitem{macharet2020adaptive}Macharet, D. G., Chen, A. K., Shishika, D., Pappas, G. J., and Kumar, V.
\newblock Adaptive partitioning for coordinated multi-agent perimeter defense
\newblock \emph{Proceedings of the International Conference on Intelligent Robots and Systems}, pp. 7971--7977, 2020.







\bibitem{pourghorban2022target}Pourghorban, A., Dorothy, M., Shishika, D., Von Moll, A., and Maity, D.
\newblock Target defense against sequentially arriving intruders
\newblock \emph{Proceedings of the 61st Conference on Decision and Control}, pp. 6594--6601, 2022.

\bibitem{pourghorban2023target}Pourghorban, A., and Maity, D.
\newblock Target defense against periodically arriving intruders
\newblock \emph{Proceedings of the American Control Conference}, pp. 1673--1679, 2023.



\bibitem{pourghorban2023targetnew}Pourghorban, A., and Maity, D.
\newblock Target defense against a sequentially arriving cooperative intruder team
\newblock \emph{Proceedings of the Open Architecture/Open Business Model Net-Centric Systems and Defense Transformation}, vol. 12544, pp. 65--77, 2023.

\bibitem{shishika2021partial}Shishika, D., Maity, D., and Dorothy, M.
\newblock Partial information target defense game
\newblock \emph{Proceedings of the International Conference on Robotics and Automation}, pp. 8111--8117, 2021.


























\bibitem{von2010market}Von Stackelberg, H.
\newblock Market structure and equilibrium
\newblock Springer Science \& Business Media, 2010.

\end{thebibliography}
\end{document}